\documentclass[a4paper,11pt,reqno]{article}

\usepackage{amssymb}
\usepackage{latexsym}
\usepackage{amsmath,hyperref,enumitem}
\usepackage{amsthm,graphicx,caption}
\usepackage{comment}
\usepackage{bm}

\usepackage{array}
\usepackage[margin=2.7cm]{geometry}

\numberwithin{equation}{section}

\newtheorem{claim}{Claim}

\newcommand*{\myproofname}{Proof}

\newcommand{\sH}{\sigma^{\textnormal{H}}}

\newcommand{\Min}{\textnormal{Min}}

\newcommand{\wt}{\textnormal{wt}}
\newcommand{\wL}{\textnormal{wt}^{\textnormal{L}}}

\newcommand{\FF}{{\mathbb F}}
\newcommand{\NN}{{\mathbb N}}

\theoremstyle{definition}
\newtheorem{theorem}{Theorem}[section]
\newtheorem{corollary}[theorem]{Corollary}
\newtheorem{proposition}[theorem]{Proposition}
\newtheorem{definition}[theorem]{Definition}
\newtheorem{example}[theorem]{Example}
\newtheorem{notation}[theorem]{Notation}
\newtheorem{remark}[theorem]{Remark}
\newtheorem{lemma}[theorem]{Lemma}

\newtheorem{thm}[theorem]{Theorem}
\newtheorem{prop}[theorem]{Proposition}

\theoremstyle{definition}
\newtheorem{defn}[theorem]{Definition}
\newtheorem{rmk}[theorem]{Remark}

\usepackage[dvipsnames]{xcolor}

\newcommand\qbin[3]{\left[\begin{matrix} #1 \\ #2 \end{matrix} \right]_{#3}}

\newcommand\eli[1]{{{\textcolor{red}{#1}}}}

\newcommand{\numberset}{\mathbb}
\newcommand{\N}{\numberset{N}}
\newcommand{\Z}{\numberset{Z}}

\newcommand{\F}{\numberset{F}}

\newcommand{\mC}{\mathcal{C}}

\newcommand{\mF}{\mathcal{F}}

\DeclareMathOperator{\lcm}{lcm}

\newcommand{\sL}{\sigma^{\textnormal{L}}}

\makeatletter
\newcommand{\mylabel}[2]{#2\def\@currentlabel{#2}\label{#1}}
\makeatother

\newcommand{\st}{\, : \, }

\usepackage{setspace}
\title{\textbf{Generalized Weights of Codes over Rings} \\ \textbf{and Invariants of Monomial Ideals}}

\author{Elisa Gorla and Alberto Ravagnani}

\date{}

\usepackage{enumitem}
\setitemize{itemsep=0pt}
\setenumerate{itemsep=0pt}

\begin{document}

\maketitle

\thispagestyle{empty}

\begin{abstract}
We develop an algebraic theory of supports for $R$-linear codes of fixed length, where $R$ is a finite commutative unitary ring. A support naturally induces a notion of generalized weights and allows one to associate a monomial ideal to a code. Our main result states that, under suitable assumptions, the generalized weights of a code can be obtained from the graded Betti numbers of its associated monomial ideal. In the case of $\FF_q$-linear codes endowed with the Hamming metric, the ideal coincides with the Stanley-Reisner ideal of the matroid associated to the code via its parity-check matrix. In this special setting, we recover the known result that the generalized weights of an $\FF_q$-linear code can be obtained from the graded Betti numbers of the ideal of the matroid associated to the code. We also study subcodes and codewords of minimal support in a code, proving that a large class of $R$-linear codes is generated by its codewords of minimal support. 
\end{abstract}

\bigskip

\section*{Introduction}\label{sec:intro}

In the past seventy years, much effort has been devoted to the study of algebraic and combinatorial objects associated to linear error-correcting codes. 
Of particular interest is the matroid associated to a linear code via its parity-check matrix, whose circuits are the minimal \textit{Hamming supports} of the codewords. Many central results in classical coding theory, including the celebrated \textit{MacWilliams identities}, their generalizations, and the duality between puncturing and shortening can be elegantly obtained via this correspondence; see e.g.~\cite{barg,britz2007higher,britz2010code,jurrius2013codes} and the references therein.

The matroid associated to a linear code via its parity check matrix retains a wealth of information about the structure of the code, including its length, dimension, minimum distance, weight distribution, and generalized weights. Moreover, the \textit{weight enumerator}
is determined by the Tutte polynomial of the matroid, see~\cite{greene1976weight}.
In addition, in~\cite{J2013} it is shown that the code's \textit{generalized weights} are determined by the \textit{graded Betti numbers} of the \textit{Stanley-Reisner ideal} of the matroid. The approach of \cite{J2013} heavily relies on matroid theory and on the properties of the Hamming support.  

In this paper, we depart from the classical theory of linear codes over a finite field and consider instead $R$-linear codes $C \subseteq R^n$, where $R$ is a finite commutative unitary ring. 
We start by proposing a general definition of support as a function $\sigma:R^n \to \N^u$ that enjoys a few natural properties. This naturally extends the notion of \textit{Hamming support} traditionally studied in coding theory~\cite[page~177]{macwilliams1977theory}. We give several examples of supports and operations to construct new supports from old. We define the support of a code $C \subseteq R^n$ as the join of the supports of its elements.

We then define the generalized weights of a code via the supports of its subcodes. Moreover, we identify a class of supports under which the algebra of the module $R^n$ is compatible with the combinatorics of the poset $\N^u$ with the product order. We call these supports \textit{modular} and establish some of their structural properties. As one might expect, the Hamming support is an example of a modular support.

Most of the paper is devoted the study of codes $C \subseteq R^n$ endowed with modular supports. We characterize their minimal codewords and prove that, if $R$ is a principal ideal ring, then their generalized weights are attained by subcodes that are minimally generated by codewords of minimal support in $C$. We also provide evidence in various examples that our results do not extend to support functions that are not modular.

The centerpiece of this paper is a result connecting the theory of modular supports with invariants of monomial ideals.
We associate a monomial ideal to a code $C\subseteq R^n$ via the supports of its codewords. Under this correspondence, inclusion of supports translates into divisibility among monomials. In Theorem~\ref{mainth} we prove that, under suitable assumptions, the generalized weights of an $R$-linear code endowed with a modular support are determined by the graded Betti numbers of the associated monomial ideal. This generalizes a result of~\cite{J2013}, with a stand-alone proof that does not rely on matroid theory. 

We conclude the paper with a series of observations on the Hamming support.
We review known results in the light of our contribution, such as the fact that every $\F_q$-linear code is generated by its codewords of minimal support. We also show that the same is true for those of maximal support, provided that $q$ is sufficiently large (and false in general for binary codes).

\medskip

\paragraph*{Outline.} In Section~\ref{sec:1} we define $R$-linear codes and review some algebra results about finite commutative rings and finite local rings.
In Section~\ref{sec:2}
we define (modular) supports and generalized weights, establishing their main properties. Codewords and subcodes of minimal supports are studied in Section~\ref{sec:3}. In Section~\ref{sec:4} we associate a monomial ideal to a code. Moreover, we prove that the generalized weights of the code are determined by the graded Betti numbers of the corresponding ideal.
We study the Hamming support in Section~\ref{sec:5}.

\medskip

\paragraph*{Acknowledgements.} We are grateful to Maria Evelina Rossi for suggesting Example \ref{marilina}.

\section{Codes and minimal systems of generators}\label{sec:1}

In this section we introduce codes over finite commutative rings and describe some properties of their systems of generators. 

\begin{notation} \label{mainnot}
Throughout the paper $n$ and $m$ denote positive integers and  $(R,+,\cdot)$ is a finite commutative ring. All rings in this paper are unitary with $1 \neq 0$. We denote by $\NN=\{0,1,2,\ldots\}$ the set of natural numbers. For $a \in \N$ we let~$[a]:=\{1,\ldots,a\}$.
\end{notation}

A classical theorem in commutative algebra states that every finite commutative ring $R$ is isomorphic to a finite product of
finite local rings. We forget the isomorphism and write
\begin{equation}\label{decR}
R=R_1 \times\ldots\times R_\ell,
\end{equation}
where $R_1,\ldots,R_\ell$ are finite local rings. 
For $i\in [\ell]$, let $\mathfrak{M}_i$ be the maximal ideal of $R_i$ and let $J=J(R)$ be the {\bf Jacobson radical} of $R$. Recall that the Jacobson radical of a commutative ring $R$ is the intersection of all maximal ideals of $R$, equivalently $$J=\{r\in R\st 1+rs \mbox{ is invertible for every } s\in R\}.$$ It is easy to check that in our situation $$J=\mathfrak{M}_1 \times\ldots\times\mathfrak{M}_\ell\subseteq R_1 \times\ldots\times R_\ell=R.$$ 
If $R$ is a finite principal ideal ring, then each $R_i$ is a finite chain ring. For $i\in[\ell]$, let $\mathfrak{M}_i=(\alpha_i)$. Then $J=(\alpha)$, where $\alpha=(\alpha_1,\ldots,\alpha_\ell)$. In particular,
if $R$ is a product of finite fields, then $\mathfrak{M}_i=0$ for $i\in[\ell]$ and $J=0$. If $R=\mathbb{F}_q$, then $\ell=1$ and $J=0$ is the only maximal ideal of $R$. 

We denote by $(R,\mathfrak{M})$ a local ring $R$ with maximal ideal $\mathfrak{M}$. If $(R,\mathfrak{M})$ is a finite local ring, then $R/\mathfrak{M}$ is a finite field.

\bigskip

In this paper we consider codes of fixed length over the alphabet $R$. All of our codes are assumed to be linear over $R$. 

\begin{definition}
An \textbf{$R$-linear code}, or simply a \textbf{code}, is an $R$-submodule $C \subseteq R^n$. The elements of $C$ are called \textbf{codewords}. A~\textbf{subcode} of $C$ is an $R$-submodule $D \subseteq C$. 
\end{definition}

Denote by $e_i$ the element of $R$ which corresponds to $(0,\ldots,0,1,0,\ldots,0)\in R_1\times\ldots\times R_\ell$, where the one appears in the $i$th component. From the decomposition in~\eqref{decR} one has
$$R^n=R_1^n \times\ldots\times R_\ell^n.$$
In the sequel, for $i \in \{1,\ldots,\ell\}$ we denote by $\pi_i:R^n\rightarrow R_i^n$ the standard projection on the $i$th coordinate.

Let $C \subseteq R^n$ be a code. For any $v=(v_1,\ldots,v_\ell)\in C$, with $v_i=\pi_i(v)\in R_i^n$, one has $$(0,\ldots,0,v_i,0,\ldots,0)=e_i v\in C.$$
Hence, up to isomorphism, $C$ can be uniquely written as 
\begin{equation}\label{code_dec}
    C=C_1\times\ldots\times C_\ell \subseteq R^n,
\end{equation}
where $C_i=\pi_i(C)\subseteq R_i^n$ for all $i \in [\ell]$.
We often consider codes $C\subseteq 0:_{R^n} J$.
Recall that
$$0:_{R^n} J=\{v\in R^n\st rv=0 \mbox{ for all } r\in J\}.$$
Then $$0:_{R^n} J=\left(0:_{R_1^n} \mathfrak{M}_1\right)\times\ldots\times\left(0:_{R_{\ell}^n} \mathfrak{M}_\ell\right).$$
Since $0:_{R^n}J$ is an $R$-module annihilated by $J$, it is an $R/J$-module. Hence, if $(R,\mathfrak{M})$ is a local ring, then $0:_{R^n}\mathfrak{M}$ is a vector space over $R/\mathfrak{M}$. 

For $C\subseteq R^n$ a code, we also consider the {\bf socle} 
$$0:_C J=\{v\in C\st rv=0 \mbox{ for all } r\in J\}=C\cap (0:_{R^n} J).$$ The socle of $C$ is a the largest subcode of $C$ which is annihilated by $J$. In particular, if $(R,\mathfrak{M})$ is a local ring, then $0:_C\mathfrak{M}$ is an $R/\mathfrak{M}$-vector space. 

A {\bf minimal system of generators} of a code $C \subseteq R^n$ is a subset of $C$ whose elements generate $C$ and which is minimal with respect to inclusion. Notice that any system of generators of a code~$C$ contains a minimal system of generators of~$C$. 

\begin{definition} \label{defmuM}
We denote by $\mu(C)$ the least cardinality of a system of generators of a code~$C$, with $\mu(0)=0$ by convention.
For a code $C=C_1 \times\ldots\times C_\ell \subseteq R^n$ as in \eqref{code_dec}, let
$$M(C):=\mu(C_1)+\ldots+\mu(C_\ell).$$
\end{definition}

\begin{example}
Let $R=R_1\times \dots \times R_\ell$, with $R_i$ a finite local ring for $i\in[\ell]$. Then $$M(R^n)=\mu(R_1^n)+\ldots+\mu(R_\ell^n)=n\ell.$$
\end{example}

Clearly $\mu(C)\leq M(C)$ for every code $C \subseteq R^n$. 
If $R$ is a finite local ring, all minimal systems of generators of a code $C \subseteq R^n$ have the same cardinality $\mu(C)=M(C)$. This  is a consequence of the next lemma, which summarizes some well-known properties of systems of generators of modules over local rings; see e.g.~\cite[Theorem 2.3]{M}. 

\begin{thm}\label{matsumura}
Let $(R,\mathfrak{M})$ be a local ring and let $C=\langle v_1,\ldots,v_t\rangle$ be an $R$-module. 
The elements $v_1,\ldots,v_t$ are a minimal system of generators of~$C$ if and only if the equivalence classes
$\overline{v_1},\ldots,\overline{v_t}$ are an $R/\mathfrak{M}$-basis of the vector space $C/\mathfrak{M}C$.
In particular, every minimal system of generators of $C$ has cardinality $\mu(C)=\dim_{R/\mathfrak{M}}(C/\mathfrak{M}C)$. 
\end{thm}

Over an arbitrary $R$, however, not all minimal system of generators of a code $C \subseteq R^n$ have the same cardinality.

\begin{example} \label{exZ6}
Let $R=\mathbb{Z}_6$ and $D=\langle (2,3)\rangle\subseteq C=\mathbb{Z}_6^2$. Then $\mu(D)=1$. Moreover, $(2,0)=4(2,3)\in D$ and $(0,3)=3(2,3)\in D$. Hence $D=\langle(2,0),(0,3)\rangle$ and $\{(2,0),(0,3)\}$ is a minimal system of generators of $D$ of cardinality~$2=M(D)$.
\end{example}

\begin{notation}
Let $C\subseteq R^n$ be a code and let
$$\mathcal{S}_{j}(C):=\{D\subseteq C \mbox{ subcode}\st \mbox{$D$ has a minimal system of generators of cardinality $j$}\}.$$ In particular, $\mathcal{S}_{j}(R^n)$ is the set of codes $C\subseteq R^n$ which have a minimal system of generators of cardinality $j$.
\end{notation}

One can show that~$M(C)$ is the largest cardinality of a minimal system of generators of $C\subseteq R^n$. 

\begin{thm}\label{lem:minspecial}
If $C\in\mathcal{S}_i(R^n)$, then there exist $v_1,\ldots,v_i$ minimal generators of $C$ with the property that $$V=\{e_jv_k\st (j,k)\in [\ell]\times [i], \; e_jv_k\neq 0\}$$ is a minimal system of generators of $C$ with $|V|\geq i$. 
Moreover $$M(C)=\max\{i\geq 0\st C\in\mathcal{S}_i(R^n)\}$$ and any minimal system of generators of $C$ of cardinality $M(C)$ has the same form as $V$.
\end{thm} 

\begin{proof}
Let $w_1,\ldots,w_i$ be a minimal system of generators of $C=C_1\times\ldots\times C_\ell$. Let $C_j'=0\times\ldots\times 0\times C_j\times 0\times\ldots\times 0\subseteq C$.
Observe that $e_jw_k\in C$ for all $k$ and $j$ and $w_k=e_1w_k+\ldots+e_\ell w_k$ for all $k\in [i]$. This proves that the set $\{e_jw_k\st (j,k)\in[\ell]\times[i],\, e_jw_k\neq 0\}$ generates $C$. Moreover, the set $\{e_jw_k\st k\in[i],\, e_jw_k\neq 0\}$ generates $C_j'$ for any $j\in[\ell]$.

Fix $j\in [\ell]$. If $e_jw_1,\ldots,e_jw_i$ do not form a minimal system of generators of $C_j'$, suppose up to reindexing that $e_jw_1,\ldots,e_jw_k$ do, for some $k<i$. For $h\in[i]\setminus[k]$, write $e_jw_h=r_{h,1}e_jw_1+\ldots+r_{h,k}e_jw_k$ for some $r_{h,1},\ldots,r_{h,k}\in R$. Let $v_h=w_h$ for $h\in[k]$, $v_h=w_h-r_{h,1}e_jw_1-\ldots-r_{h,k}e_jw_k$ for $h\in[i]\setminus[k]$. Then $v_1,\ldots,v_i$ are a minimal system of generators of $C$ with the property that $e_jv_1,\ldots,e_jv_k$ are a minimal system of generators of $C_j'$ and $e_jv_{k+1}=\ldots=e_jv_i=0$. 
Notice that only the $j$th coordinate of $v_1,\ldots,v_i$ was affected by this operation, hence $e_hw_k=e_hv_k$ for all $k\in[i]$ if $h\neq j$. Performing this operation for all $j\in [\ell]$ produces a minimal system of generators $v_1,\ldots,v_i$ of $C$ with the property that the set $V=\{e_jv_k\st (j,k)\in [\ell]\times [i], \, e_jv_k\neq 0\}$ is a minimal system of generators of $C$. Moreover, for $j\in[\ell]$, $\{e_jv_k\st k\in[i], \, e_jv_k\neq 0\}$ is a minimal system of generators of $C_j'$. Since $v_1,\ldots,v_i\neq 0$, for each $k\in[i]$ there must be at least a $j\in[\ell]$ such that $e_jv_k\neq 0$. This proves that $|V|\geq i$.

To prove the last part of the statement, let $M=\max\{i\geq 0\st C\in\mathcal{S}_i(R^n)\}.$ Since $C$ has a minimal system of generators $v_1,\ldots,v_M$, by the first part of the statement $V=\{e_jv_k\st (j,k)\in [\ell]\times [M], \, e_jv_k\neq 0\}$ is a minimal system of generators of $C$ of $|V|=M$. Therefore, for each $k\in[M]$ there is exactly one $j\in[\ell]$ with $e_jv_k\neq 0$. Moreover, $\{v_k\st k\in[M],\, e_jv_k\neq 0\}$ is a minimal system of generators of $C_j'$ for $j\in[\ell]$, hence it has cardinality $\mu(C_j)$ by Theorem~\ref{matsumura}. It follows that $M=\mu(C_1)+\ldots+\mu(C_\ell)=M(C)$. 
\end{proof}

We conclude the section with a few elementary properties of $M(C)$.

\begin{proposition}\label{lem:equal}
Let $R$ be a finite commutative ring and let $D\subseteq C\subseteq 0:_{R^n}J$ be codes. Then $$M(D)\leq M(C)$$ and equality holds if and only if $D=C$.
\end{proposition}

\begin{proof}
Write $C=C_1\times\ldots\times C_\ell$ and $D=D_1\times\ldots\times D_\ell$. Since $D\subseteq C$, we have $D_i\subseteq C_i\subseteq 0:_{R_i^n}\mathfrak{M}_i$ for all $i\in [\ell]$. So $C_i$ and $D_i$ are $R_i/\mathfrak{M}_i$-vector spaces and  $\mu(D_i)\leq \mu(C_i)$ for all $i\in [\ell]$ by Theorem~\ref{matsumura}. It follows that $$M(D)=\mu(D_1)+\ldots+\mu(D_\ell)\leq \mu(C_1)+\ldots+\mu(C_\ell)=M(C).$$ Moreover, $M(D)=M(C)$ if and only if $\mu(D_i)=\mu(C_i)$ for all $i\in [\ell]$. In this case, $C_i$ and $D_i$ are $R_i/\mathfrak{M}_i$-vector spaces
of the same dimension by Theorem~\ref{matsumura}. Hence they are equal, therefore $D=C$.
\end{proof}

Notice that one may have $D\subsetneq C\subseteq R^n$ with $M(D)=M(C)$. Some examples of this arise for instance from the fact that, over a principal ideal ring (\textbf{PIR} in the sequel), the value of $M(C)$ does not change when replacing $C$ with its socle. We will use this fact repeatedly throughout the paper.

\begin{proposition}\label{ann}
Let $R$ be a finite PIR and let $C\subseteq R^n$ be a code. Then $$M(C)=M(0:_C J).$$
\end{proposition}

\begin{proof}
We may assume without loss of generality that $R$ is a finite chain ring. Indeed, if the result is true for finite chain rings, then write
$R=R_1\times\ldots\times R_\ell$ as a product of finite chain rings and $C=C_1\times\ldots\times C_\ell$, where $C_i\subseteq R_i^n$ is a code for $i\in [\ell]$. We have $$M(C)=\mu(C_1)+\ldots+\mu(C_\ell)=\mu(0:_{C_1}\mathfrak{M}_1)+\ldots+\mu(0:_{C_\ell}\mathfrak{M}_\ell)=M(0:_C J),$$ where the last equality follows from $$0:_C J=(0:_{C_1}\mathfrak{M}_1)\times\ldots\times(0:_{C_\ell}\mathfrak{M}_\ell).$$
In order to prove that $\mu(C)=\mu(0:_C J)$ for $C\subseteq R^n$ and $R$ a finite chain ring, observe that $J=(\alpha)$ is principal and 
\begin{equation}\label{eq:dims}
\mu(C)=\dim_{R/(\alpha)}(C/\alpha C)=
\dim_{R/(\alpha)}(0:_C \alpha)=\mu(0:_C \alpha),
\end{equation}
where the first and last equalities follow from Theorem~\ref{matsumura}.
The short exact sequence 
$$0\to 0:_C\alpha\to C\to\alpha C\to 0$$
induces an isomorphism $C/\alpha C\cong 0:_C\alpha$, which proves the central equality in (\ref{eq:dims}).
\end{proof}

The statement of Proposition \ref{ann} also holds when $C=R^n$ and $R$ is a product of finite Gorenstein local rings. 

\begin{example}\label{ex:gorenstein}
Write $R=R_1\times\ldots\times R_\ell$ and suppose that each $R_i$ is a finite Gorenstein local ring. Suppose first that $\ell=1$, i.e. $R$ is a Gorenstein local ring with maximal ideal $\mathfrak{M}$. We have the following isomorphisms of $R/\mathfrak{M}$-vector spaces $$R^n/\mathfrak{M}R^n\cong(R/\mathfrak{M}R)^n\cong (0:_R \mathfrak{M})^n= 0:_{R^n}\mathfrak{M},$$ 
where the central isomorphism follows from the definition of a Gorenstein local ring. Then $\mu(R^n)=\mu(0:_{R^n}\mathfrak{M})=n$ by Theorem \ref{matsumura}.

For general $\ell$, one has $$M(R^n)=\mu(R_1^n)+\ldots+\mu(R_\ell^n).$$ Moreover, $0:_{R^n} J=\left(0:_{R_1^n} \mathfrak{M}_1\right)\times\ldots\times\left(0:_{R_{\ell}^n} \mathfrak{M}_\ell\right)$, hence $$M(0:_{R^n} J)=\mu\left(0:_{R_1^n} \mathfrak{M}_1\right)+\ldots+\mu\left(0:_{R_{\ell}^n} \mathfrak{M}_\ell\right).$$
It follows from the previous case ($\ell=1$) that $$\mu(R_i^n)=\mu(0:_{R_i^n}\mathfrak{M})=n$$ for all $i\in[\ell].$
\end{example}

The argument of Example \ref{ex:gorenstein} also shows that, if the equality in Proposition \ref{ann} holds for  $C=R^n$, for $R$ a finite local ring, then $R$ must be Gorenstein. However, Proposition \ref{ann} is not true in general over finite Gorenstein local rings. The next example illustrating this was suggested to us by Maria Evelina Rossi.

\begin{example}\label{marilina}
Let $R=\F_2[x,y]/(x^2,y^2)$. Then $R$ is a finite local ring with maximal ideal $\mathfrak{M}=(x,y)$. Let $C=\mathfrak{M}$. Then $\mu(C)=2$, but $\mu(0:_C \mathfrak{M})=\mu(\langle xy\rangle)=1$.
\end{example}

\section{Supports and generalized weights}\label{sec:2}

In this section we develop an algebraic theory of supports over a finite commutative ring~$R$. We propose a general definition of support as a map $R^n \to \N^u$, which 
naturally induces a notion of generalized weights for codes $C \subseteq R^n$.
This extends the notion of generalized Hamming weights for codes that are linear over a finite field $\FF_q$. We establish some properties of support functions and generalized weights.
We also define a family of supports, the \textit{modular supports}, whose associated generalized weights will be studied in the next sections.

\begin{notation}
In the sequel, $u \ge 1$ is an integer.
For $s,t \in \NN^u$ write $s \leq t$ if $s_i \le t_i$ for all $i \in [u]$. Then $(\NN^u, \leq)$ is a (poset) lattice. The {\bf meet} of $s,t \in \NN^u$
is the element $s \wedge t \in \NN^u$ given by $(s \wedge t)_i=\min\{s_i,t_i\}$ for all $i \in [u]$. The {\bf join}
of $s,t \in \NN^u$, denoted by $s \vee t$, is $(s \vee t)_i=\max\{s_i,t_i\}$ for all $i \in [u]$. For $s \in \N^u$, we let $|s|:=s_1+ \cdots +s_u$. \end{notation}

\begin{definition} \label{defsupp}
A \textbf{support} on $R^n$ is a function $\sigma:R^n \to \NN^u$ with the following properties.
\begin{enumerate}[label={(\arabic*)}]
\item[\mylabel{P1}{(P1)}] $\sigma(v)=0$ if and only if $v=0$.
\item[\mylabel{P2}{(P2)}] $\sigma(rv) \le \sigma(v)$ for all $r \in R$ and $v \in R^n$.
\item[\mylabel{P3}{(P3)}] $\sigma(v+w) \le \sigma(v) \vee \sigma(w)$ for all $v,w \in R^n$.
\end{enumerate}
\end{definition}

A support function $\sigma:R^n \to \N^u$ satisfies the following additional properties.

\begin{lemma} \label{prsupp}
Let $\sigma:R^n \to \N^u$ be a support. The following hold.
\begin{enumerate}[label={(\arabic*)}]
\item[\mylabel{prsupp1}{(P4)}] If $v \in R^n$ and $r \in R$ is a unit, then $\sigma(rv)=\sigma(v)$.
\item[\mylabel{prsupp2}{(P5)}] If $v,w \in R^n$ and $i \in [u]$ satisfy $\sigma(v)_i=0$ and $\sigma(w)_i \neq 0$, then $\sigma(w+v)_i \neq 0$.
\item[\mylabel{prsupp3}{(P6)}] If $v,w \in R^n$ and $i \in [u]$ satisfy $\sigma(v)_i=\sigma(w)_i=0$, then $\sigma(w+v)_i = 0$.
\end{enumerate}
\end{lemma}

\begin{proof}
The first claim easily follows from Property~\ref{P2}.
To see the second, suppose towards a contradiction that $\sigma(w+v)_i=0$. Since $w=(w+v)+(-v)$, by~\ref{P3} and the first claim we have $\sigma(w) \le \sigma(w+v) \vee \sigma(-v) = \sigma(w+v) \vee \sigma(v)$, from which $\sigma(w)_i=0$, a contradiction.
Finally, the third claim follows from Property~\ref{P3}.
\end{proof}

A support $\sigma:R^n \to \N^u$ naturally induces a weight function $\wt:R^n \to \N$. 

\begin{definition}
The \textbf{weight} of $v \in R^n$ is $\wt(v):=|\sigma(v)|$.
The {\bf minimum weight} and {\bf maximum weight} of a code $0\neq C\subseteq R^n$ are, respectively,
$$\min\wt(C):=\min\{\wt(v)\st v\in C\setminus\{0\}\} \qquad \mbox{and} \qquad 
\max\wt(C):=\max\{\wt(v)\st v\in C\}.$$
\end{definition}

Notice that the function $\wt:R^n\to\NN$ has indeed the properties of a weight, since $\wt(v)\geq 0$ for all $v$, $\wt(v)=0$ if and only if $v=0$, and $\wt(u+v)\leq \wt(u)+\wt(v)$  for all $u,v\in R^n$. In addition, the weight satisfies $\wt(rv)\leq \wt(v)$ for all $r\in R$ and $v\in R^n$ and $\wt(rv)=\wt(v)$ for $r\in R$ invertible and $v\in R^n$. 

We give some examples of support functions. Many others can be obtained by applying Proposition~\ref{prop:ext} to these examples. 

\begin{example} \label{exasupp}
\begin{enumerate}[label={(\arabic*)}]
\item \label{exasupp3} The function $\sigma:\F_2^2 \to \{0,1,2\}^3$ defined by $\sigma(0,0)=(0,0,0)$, $\sigma(1,0)=(2, 0, 2)$, $\sigma(0,1)=(2, 1, 0)$, and
$\sigma(1,1)=(0,1,2)$ is a support.
\item \label{exasupp2}  Let $R$ a be finite ring and let $0=I_0 \subsetneq I_1 \subsetneq
\cdots \subsetneq I_{\epsilon-1} \subsetneq I_{\epsilon}=R$
be a chain of ideals of $R$. For $r \in R$, let $\sigma(r):=\min\{0 \le i \le \epsilon \st r \in I_i\}$. Extend $\sigma$ coordinatewise to $\sigma:R^n \to \{0,\ldots,\epsilon\}^n$. It can be checked that $\sigma$ is a support, called the \textbf{chain support} on $R$, see~\cite[Example 26]{ravagnani2018duality}.

\item \label{exasupp2a} Let $R$ be a finite chain ring. The chain support associated to the full chain of ideals of $R$ is the \textbf{chain ring support} on $R$.
\item \label{exasupp1} If $R=\F_q$ is a finite field, the chain ring support coincides with the \textbf{Hamming support} $\sH:\F_q^n \to \{0,1\}^n$, given by $\sH(v)_i=1$ if $v_i \neq 0$ and $\sH(v)_i=0$ if $v_i=0$. See~\cite{macwilliams1977theory} for a general reference on Hamming-metric codes.

\item In his master thesis \cite{gassner20}, written under the direction of J. Rosenthal and V. Weger, N.~Gassner introduces the $p$-adic weight and distance on $\mathbb{Z}_{p^e}^n$, where $p$ is prime and $e\geq 1$. The $p$-adic weight on $\mathbb{Z}_{p^e}$ induces the same partition as the weight associated to the chain ring support of $\mathbb{Z}_{p^e}$.
\end{enumerate}
\end{example}

Not all the supports studied in the coding theory literature are supports according to Definition~\ref{defsupp}.

\begin{example}\label{rem:lee}
The \textbf{Lee weight} $\wL:\Z_4 \to \{0,1,2\}$ is defined by $\wL(0)=0$, $\wL(1)=\wL(3)=1$ and $\wL(2)=2$. Its coordinatewise extension to $\Z_4^n$ is not a support in the sense of Definition~\ref{defsupp}. For instance, $\wL(1+1)=2 \not\le \max\{\wL(1), \wL(1)\}=1$, contradicting Property~\ref{P3}.
\end{example}

In the next proposition we list some simple operations that allow one to construct new supports from known ones. The proof is straightforward and left to the reader.

\begin{proposition}\label{prop:ext}
\begin{enumerate}[label={(\arabic*)}]
\item Let $\sigma:R^n \to \N^u$ be a function and let $s:\N^u\to \N^u$ be a permutation of the coordinates. Then $\sigma$ is a support if and only if $s\circ\sigma$ is a support.
\item \label{prop:ext1} Let $\sigma_i:R^{n_i} \to \N^{u_i}$ be functions, $i\in[\ell]$. Let $n=n_1+\ldots+n_\ell$ and $u=u_1+\ldots+u_\ell$. Let
$$\sigma=\sigma_1 \times \cdots \times \sigma_\ell:R^n \to \N^u, \qquad (v_1,\ldots,v_\ell) \mapsto (\sigma_1(v_1),\ldots,\sigma_\ell(v_\ell)).$$ Then $\sigma$ is a support if and only if $\sigma_1,\ldots,\sigma_\ell$ are supports.
\item \label{prop:ext2} Let $\sigma:R^n \to \N^u$, $\sigma(v)=(\sigma_1(v),\ldots,\sigma_u(v))$, be a function. Let $i\in [u]$, $a\in\mathbb{N}\setminus\{0\}$, and define $$\sigma_{i,a}:R^n \to \N^u,\qquad v\mapsto (\sigma_1(v),\ldots,\sigma_{i-1}(v),a\sigma_i(v),\sigma_{i+1}(v),\ldots,\sigma_u(v)).$$ Then $\sigma$ is a support if and only if $\sigma_{i,a}$ is.
\item \label{prop:ext3} Let $\sigma:R^n \to \N^u$, $\sigma(v)=(\sigma_1(v),\ldots,\sigma_u(v))$. For $i\in[u]$, let $$\tilde{\sigma}_i:R^n \to \N^{u+1},\qquad v\mapsto (\sigma_1(v),\ldots,\sigma_{i-1}(v),\sigma_i(v),\sigma_i(v),\sigma_{i+1}(v),\ldots,\sigma_u(v)).$$ Then $\sigma$ is a support if and only if $\tilde{\sigma}_{i}$ is.
\item Let $\sigma:R^n \to \N^u$, $\sigma(v)=(\sigma_1(v),\ldots,\sigma_u(v))$, be a support and let $i\in[u]$. Assume that there are no $a\in\N\setminus\{0\}$ and $v\in R^n$ such that $\sigma(v)=(0,\ldots,0,a,0,\ldots,0)$, where $a$ appears in the $i$th entry. Then $$\hat{\sigma}_i:R^n \to \N^{u-1},\qquad v\mapsto (\sigma_1(v),\ldots,\sigma_{i-1}(v),\sigma_{i+1}(v),\ldots,\sigma_u(v))$$ is a support.
\item Let $\sigma:R^n \to \N^u$ be a function, $\sigma(v)=(\sigma_1(v),\ldots,\sigma_u(v))$. For $i\in[u]$, let $$\check{\sigma}_i:R^n \to \N^{u+1},\qquad v\mapsto (\sigma_1(v),\ldots,\sigma_{i-1}(v),\sigma_i(v),0,\sigma_{i+1}(v),\ldots,\sigma_u(v)).$$ Then $\sigma$ is a support if and only if $\check{\sigma}_{i}$ is.
\item \label{prop:ext4} Let $\sigma: R^n \to \N^u$ be a support and let $f: R^k \to R^n$ be an injective $R$-linear map. Then $\sigma \circ f: R^k \to \N^u$ is a support.
\item Let $\sigma_i:R^n \to \N^{u_i}$ be supports, $i\in[k]$. Then $\sigma=(\sigma_1,\ldots,\sigma_k):R^n\to\N^{u_1+\ldots+u_k}$ is a support.
\end{enumerate}
\end{proposition}

Similarly to the situation of linear codes endowed with the Hamming support, a support function over a finite commutative ring $R$ induces a notion of support of a code. In turn, this allows us to define generalized weights for $R$-linear codes. 

\begin{definition}
The \textbf{support} of a code $C \subseteq R^n$ is
$$\sigma(C):= \bigvee_{v \in C} \sigma(v)\in \NN^u.$$
\end{definition}

Notice that the support of a code is determined by the supports of the vectors in any system of generators. More precisely, let $C=\langle v_1,\ldots,v_k\rangle\subseteq R^n$. Since, by Definition~\ref{defsupp}, 
$$\sigma(r_1v_1+\ldots+r_kv_k)\leq\sigma(r_1v_1)\vee\ldots\vee\sigma(r_kv_k)\leq \sigma(v_1)\vee\ldots\vee\sigma(v_k)$$
for any $r_1,\ldots,r_k\in R$, we have \begin{equation}\label{eqn:supp}
\sigma(C)=\bigvee_{i=1}^k\sigma(v_i).
\end{equation}
Moreover, if $D\subseteq C$, then by definition $\sigma(D)\leq\sigma(C)$.

\begin{example}
Equation~\eqref{eqn:supp} does not hold for the Lee weight $\wL:\Z_4 \to \{0,1,2\}$, nor for its coordinatewise extension to $\mathbb{Z}_4^n$. For example, $\wL(\langle 1\rangle)=2>1=\wL(1)$, see also Example~\ref{rem:lee}.
\end{example}

\begin{definition} \label{def:gw}
For $r\in[M(C)]$, the $r$-th \textbf{generalized weight} of~$C$ is the integer
$$d_r(C):=\min\{|\sigma(D)| \st D\in\mathcal{S}_{j}(C) \mbox{ for some } j\geq r\}.$$
We also set $$d_0(C):=0.$$
\end{definition} 

It follows from Theorem \ref{lem:minspecial} that for $r\in [M(C)]$ we have $\mathcal{S}_{r}(C)\neq\emptyset$. Hence $d_r(C)$ is well-defined. 

\begin{rmk}\label{rmk:containment}
For $r\in[M(C)]$ one has $$d_r(C)=\min\{|\sigma(D)| \st D\in\mathcal{S}_r(C)\}.$$ Indeed, let $j\geq r$ and let $D\in\mathcal{S}_{j}(C)$. Then there exists a $D'\subseteq D$ such that $D'\in\mathcal{S}_r(C)$ and $|\sigma(D')|\leq|\sigma(D)|$.
\end{rmk}

In the next lemma we collect a few easy consequences of Definition~\ref{def:gw}. 

\begin{lemma}\label{lem:genwgt}
Let $D\subseteq C\subseteq R^n$ be codes. The following hold.
\begin{enumerate}[label={(\arabic*)}]
\item \label{lem:genwgt1} $d_1(C)=\min\wt(C)$.
\item \label{lem:genwgt2} $d_r(D)\geq d_r(C)$ for $r\in[\min\{M(D),M(C)\}]$.
\item \label{lem:genwgt2a} $d_{r+1}(C)\geq d_r(C)$ for $r\in[M(C)-1]$.
\item \label{lem:genwgt3} $d_r(C) = \min\{|\sigma(D)| \st D \subseteq C, \, M(D) \ge r\}$ for $r\in[M(C)]$.
\end{enumerate}
\end{lemma}
 
\begin{proof}
By Property~\ref{P2} one has $\sigma(\langle v\rangle)=\sigma(v)$ for any $v\in C$. Hence \ref{lem:genwgt1}
follows, thanks to Remark \ref{rmk:containment}. Part \ref{lem:genwgt2} holds since every subcode of $D$ is also a subcode of $C$. Part \ref{lem:genwgt2a} follows from observing that $d_r$ is the minimum of the function $|\sigma(D)|$ for $D$ ranging over the set $\mathcal{S}_r(C)\cup\ldots\cup\mathcal{S}_{M(C)}(C)$ and passing from $r$ to $r+1$ we minimize over a subset.
In order to prove \ref{lem:genwgt3}, let $i\geq r$ and $D\in\mathcal{S}_i(C)$.
Then $M(D)\geq i$ by Theorem~\ref{lem:minspecial}. Therefore, if $D\in\mathcal{S}_i(C)$ for some $i\geq r$, then $D\in\mathcal{S}_{M(D)}(C)$ and $M(D)\geq r$.
Since every $D\subseteq C$ belongs to $\mathcal{S}_{M(D)}(C)$, then $\{D\in\mathcal{S}_i(C)\mbox{ for some } i\ge r\}=\{D \subseteq C, \, M(D) \ge r\}$.
Therefore \begin{equation*}
d_r(C) = \min\{|\sigma(D)| \st D\in\mathcal{S}_i(C)\mbox{ for some } i\ge r\}=\min\{|\sigma(D)| \st D\subseteq C, \, M(D) \ge r\}.\qedhere
\end{equation*}
\end{proof}

We now show how the structure of supports relate to the decomposition of $R$ in~\eqref{decR}. 

\begin{proposition}\label{thm:supp}
Let $\sigma:R^n\rightarrow\mathbb{N}^u$ be a support. Then for any $v=(v_1,\ldots,v_\ell)\in R^n= R_1^n\times\ldots\times R_\ell^n$ we have $$\sigma(v)=\sigma_1(v_1)\vee\ldots\vee\sigma_\ell(v_\ell),$$ where $\sigma_i:R_i^n\rightarrow\mathbb{N}^u$ is as support defined via $\sigma_i(v_i):=\sigma(e_iv)$ for all $i\in[\ell]$.
\end{proposition}

\begin{proof}
It is easy to check that $\sigma_i$ is well-defined and is a support for all $i\in[\ell]$. One has 
$\sigma_i(v_i)=\sigma(e_iv)\leq\sigma(v)$, hence $\sigma_1(v_1)\vee\ldots\vee\sigma_\ell(v_\ell) \le \sigma(v)$.
Furthermore,
$$\sigma(v)=\sigma\left(\sum_{i=1}^\ell e_iv\right)\leq\bigvee_{i=1}^\ell\sigma(e_iv)=\bigvee_{i=1}^\ell\sigma_i(v_i).$$
It follows that $\sigma(v)=\sigma_1(v_1)\vee\ldots\vee\sigma_\ell(v_\ell)$,
as desired.
\end{proof}

\begin{remark}\label{rem:suppfcr}
When $R$ is a finite chain ring, support functions on $R$ have a simple description. To see this, let $\alpha$ be a generator of the maximal ideal of $R$ and let $\epsilon=\min\{i>0\st \alpha^i=0\}$. Let $\sigma,\tau:R\rightarrow\mathbb{N}^u$ be supports. Then $\sigma=\tau$ if and only if $\sigma(\alpha^i)=\tau(\alpha^i)$ for $0\leq i\leq \epsilon-1$. Indeed, every element of $R\setminus\{0\}$ is of the form $r\alpha^i$ where $r$ is a unit and $0\leq i\leq \epsilon-1$, and $\sigma(r\alpha^i)=\sigma(\alpha^i)$. Therefore, a support $\sigma:R\rightarrow\mathbb{N}^u$ corresponds to a set of vectors $a^{(0)}, a^{(1)},\ldots, a^{(\epsilon-1)}\in\mathbb{N}^u$ with the property that $a^{(0)}\geq a^{(1)}\geq\ldots\geq a^{(\epsilon-1)}$. The correspondence is determined by setting $\sigma(\alpha^i)=a^{(i)}$ for all $i\in\{0,\ldots,\epsilon-1\}$.
In particular, any support on $R$ induces the same partition as a chain support.
\end{remark}

\subsection{Modular Supports}

In this subsection we define and study a class of supports whose structure is closely related to the $R$-module structure of $R^n$, and that we therefore call \textit{modular}. This paper is primarily devoted to the study of generalized weights associated to modular supports.

\begin{definition}
A support $\sigma$ is \textbf{modular} if it satisfies the following:
\begin{enumerate}[label={(\arabic*)}]
\item[\mylabel{PP}{(P$\star$)}]  
If $v,w\in R^n$ and $i\in [u]$ satisfy $0 \neq \sigma(v)_i \le \sigma(w)_i$, then there exists $r\in R$ such that $\sigma(v+rw)_i<\sigma(v)_i$.
\end{enumerate}
\end{definition}

\begin{remark}
By repeatedly applying Property~\ref{PP}, one obtains the following equivalent property: \ 
If $v,w\in R^n$ and $i\in [u]$ satisfy $0 \neq \sigma(v)_i \le \sigma(w)_i$, then there exists $r\in R$ such that $\sigma(v+rw)_i=0$.
\end{remark}

As for supports, one can easily produce new modular supports from known ones.

\begin{proposition}\label{newPP}
Let $\sigma:R^n\rightarrow\NN^u$ be a support.
Following the notation and numbering of Proposition~\ref{prop:ext}, we have:
\begin{enumerate}[label={(\arabic*)}]
\item \label{newPP1} $\sigma$ is modular if and only if $s\circ\sigma$ is modular;
\item \label{newPP2} $\sigma_1,\ldots,\sigma_\ell$ are modular if and only if $\sigma$ is modular;
\item \label{newPP3} $\sigma$ is modular if and only if $\sigma_{i,a}$ is modular;
\item \label{newPP4} $\sigma$ is modular if and only if $\tilde{\sigma}_{i}$ is modular;
\item \label{newPP5} if $\sigma$ is modular, then $\hat{\sigma}_{i}$ is modular;
\item \label{newPP6} $\sigma$ is modular if and only if $\check{\sigma}_{i}$ is modular;
\item \label{newPP7} if $\sigma$ is modular, then $\sigma \circ f$ is modular;
\item \label{newPP8} if $\sigma_1,\ldots,\sigma_k$ are modular, then $\sigma=(\sigma_1,\ldots,\sigma_k)$ is modular.
\end{enumerate}
\end{proposition}

Several, but not all, of the supports that we have encountered so far are modular.

\begin{example}\label{ex:chain*}
Support \ref{exasupp3} of Example~\ref{exasupp} is modular, while an arbitrary chain support is not. Over a finite chain ring, the only modular chain support is the chain ring support.
For example, the chain support on $\Z_4$ associated with the chain  $0\subsetneq\Z_4$ is not modular. Indeed, $\sigma(2)=\sigma(4)=1$, but there is no $r \in \Z_4$ with $2-4r=0$.
\end{example}

The Hamming support is an example of modular support.

\begin{example}
It is easy to check that the chain ring support of Example~\ref{exasupp}\ref{exasupp2a} is modular.
Hence a product of chain ring supports is modular by Proposition~\ref{prop:ext} and Proposition~\ref{newPP}\ref{newPP2}.
In particular, the Hamming support is modular.
\end{example}

\begin{example}
The supports of Remark~\ref{rem:suppfcr} are modular if and only if $(a^{(j)})_{i}\neq (a^{(k)})_{i}$ for all $j,k\in\{0,\ldots,\epsilon-1\}$ distinct and $i\in[u]$.
\end{example}

We now give more examples of supports which are not modular.

\begin{example}\label{expato}
Let $R=\F_2$ and let 
$\sigma: \F_2^2 \to \{0,1\}^2$
be defined by
$$\sigma(0,0)=(0,0), \quad
\sigma(1,0)=(1,1), \quad
\sigma(0,1)=(0,1), \quad
\sigma(1,1)=(1,1).$$
Then $\sigma$ is a support which is not modular.
\end{example}
 
\begin{example}
The chain support on $\Z_6$ associated with the chain  $(0) \subsetneq (2) \subsetneq \Z_6$ is not modular. Indeed, $\sigma(2)=1$ and $\sigma(3)=2$, but there is no $r \in \Z_6$ with $2-3r=0$.
\end{example} 

The next result shows that every modular support over a finite commutative ring decomposes as a product of modular supports over finite local rings.

\begin{theorem}\label{thm:supp*}
Let $R$ be a finite commutative ring and let $\sigma:R^n\rightarrow\mathbb{N}^u$ be a modular support. Up to a permutation of the coordinates of $\mathbb{N}^u$ we have $\sigma=\sigma_1\times\ldots\times\sigma_\ell$ where $\sigma_i:R_i^n\rightarrow\mathbb{N}^{u_i}$ for $i\in[\ell]$ and $u_1,...,u_\ell$ are integers with $u_1+\ldots+u_\ell=u$. Moreover, $\sigma_i$ is a modular support for all $i\in[\ell]$.
\end{theorem}

\begin{proof}
For $v\in R^n$, write $v=(v_1,\ldots,v_\ell)$ with $v_i\in R_i^n$. By Proposition~\ref{thm:supp} we have $$\sigma(v_1,\ldots,v_\ell)=\sigma_1(v_1)\vee\ldots\vee\sigma_\ell(v_\ell),$$ where $\sigma_i:R_i^n\rightarrow\mathbb{N}^u$ is a support defined via $\sigma_i(v_i):=\sigma(e_iv)$, $i\in[\ell]$. We claim that for each $x\in [u]$ there is at most one $i\in[\ell]$ such that $\sigma_i(v_i)_x\neq 0$ for some $v\in R^n$. 
Indeed, assume towards a contradiction that there exist $i\neq j$ and $v,w\in R^n$ such that $\sigma_i(v_i)_x,\sigma_j(w_j)_x\neq 0$. Without loss of generality we may assume that $0<\sigma_i(v_i)_x\leq \sigma_j(w_j)_x$.
By Property~\ref{PP} there exists $r=(r_1,\ldots,r_\ell)\in R$ such that $$\sigma(e_iv)_x>\sigma(e_iv+e_jrw)_x=[\sigma_i(v_i)\vee\sigma_j(r_jw_j)]_x\geq \sigma_i(v_i)_x,$$ where the equality follows from Proposition~\ref{thm:supp}. This a contradiction, establishing the claim.

We have shown that for each $x\in [u]$ there exists at most one $i\in[\ell]$ for which $(\sigma_i)_x$ is not the zero function. In other words, the supports of the images of the functions $\sigma_i$ are disjoint. Up to permuting the coordinates of $\mathbb{N}^u$, one may assume that $\sigma_1$ is supported on the first $u_1$ coordinates, $\sigma_2$ on the next $u_2$,\ldots, and~$\sigma_\ell$ on the last $u_\ell$. Therefore one may regard each $\sigma_i$ as a function which takes values in $\mathbb{N}^{u_i}$. Then $\sigma=\sigma_1\times\ldots\times\sigma_\ell$ and each $\sigma_i$ is a modular support by Proposition \ref{prop:ext}\ref{prop:ext1} and Proposition \ref{newPP}\ref{newPP2}.
\end{proof}

By combining Remark~\ref{rem:suppfcr} with Theorem~\ref{thm:supp*}, support functions on a principal ideal ring~$R$ can be easily characterized as follows. 

\begin{corollary}\label{cor:supponR}
Let $R$ be a finite principal ideal ring and let $\sigma:R\rightarrow\mathbb{N}^u$ be a modular support. By the Zariski-Samuel Theorem, $R=R_1\times\ldots\times R_\ell$ where $R_1,\ldots,R_\ell$ are finite chain rings. For each $i$, let $\alpha_i$ be a generator of the maximal ideal of $R_i$ and let $\epsilon_i:=\min\{j\st \alpha_i^j=0\}$.
Then there exist $u_1,\ldots,u_\ell$ such that $u_1+\ldots+u_\ell=u$ and $\sigma=\sigma_1\times\ldots\times\sigma_\ell$, where $\sigma_i:R_i\rightarrow\mathbb{N}^{u_i}$ for $i\in[\ell]$.
Let $\sigma_i(\alpha_i^j)=a^{(i,j)}\in\mathbb{N}^{u_i}$ for $i\in[\ell]$ and $j\in\{0,\ldots,\epsilon_i-1\}$. Then $(a^{(i,j-1)})_k>(a^{(i,j)})_k$ for $j\in[\epsilon_i-1]$, $i\in[\ell]$, $k\in[u_i]$.

Conversely, any set of vectors $a^{(i,j)}\in\mathbb{N}^{u_i}$ such that $a^{(i,j-1)}\geq a^{(i,j)}$ for $j\in[\epsilon_i-1]$ and $i\in[\ell]$ defines a support $\sigma=\sigma_1\times\ldots\times\sigma_\ell$ on $R$ via
$\sigma_i(r\alpha_i^j)=a^{(i,j)}$ for $i\in[\ell],j\in\{0,\ldots,\epsilon_i-1\}$, and $r\in R_i$ invertible. Moreover, if $(a^{(i,j-1)})_k>(a^{(i,j)})_k$ for $j\in[\epsilon_i-1]$, $i\in[\ell]$, and $k\in[u_i]$, then~$\sigma$ is modular.
\end{corollary}

The following is a reformulation of Property~\ref{PP} for elements of $0:_{R^n}J$.

\begin{corollary}\label{lemadd}
Assume that $\sigma$ is modular. 
If $v,w \in 0:_{R^n}J$ and $i \in [u]$ satisfy $\sigma(v)_i \neq 0$ and $\sigma(w)_i \neq 0$, then there exists a unit $r \in R$ with $\sigma(v-rw)_i=0$.
\end{corollary}

\begin{proof}
By Theorem~\ref{thm:supp*} we may assume without loss of generality that $R$ is a finite local ring. Indeed, let $j\in[\ell]$ be such that $(\sigma_j)_i$ is not identically zero and suppose that $\sigma_j(v_j-r_jw_j)_i=0$ for some $r_j\in R_j$ invertible, then $r=(1,\ldots,1,r_j,1,\ldots,1)\in R$ is invertible and satisfies $\sigma(v-rw)_i=0$.

Assume now that $(R,\mathfrak{M})$ is a finite local ring.
If $\sigma(v)_i\leq\sigma(w)_i$, then by Property~\ref{PP} there is $r\in R$ such that $\sigma(v-rw)_i=0$. If $r\in\mathfrak{M}$, then $rw=0$, hence $\sigma(v)_i=0$, contradicting the assumption in the statement. Therefore $r$ is invertible. Similarly, if $\sigma(w)_i\leq\sigma(v)_i$, then there exists $s\in R$ invertible such that $0=\sigma(w-sv)_i=\sigma(v-s^{-1}w)_i$.
\end{proof}

\section{Codewords  and subcodes of minimal support}\label{sec:3}

In this section we study the codewords and subcodes of minimal support of an $R$-linear code endowed with a modular support. In particular, we establish some properties of the systems of generators of subcodes of minimal support. This allows us to derive properties of the generalized weights, such as monotonicity and a generalization of the Singleton bound.

In the sequel, we follow the notation of the previous sections and let 
$\sigma:R^n \to \N^u$ be a modular support. 
The minimal codewords of a code play a central role in our work. They are defined as follows.

\begin{definition}
Let $C \subseteq R^n$ be a code. We say that $v \in C\setminus\{0\}$ is \textbf{minimal} in $C$ if its support is minimal among the supports of the elements of $C\setminus\{0\}$.
We denote by $\Min(C)$ the set of minimal codewords of $C$.
\end{definition}

\begin{remark}
By definition, $C=0$ has no minimal codewords, i.e., $\Min(0)=\emptyset$.
\end{remark}

We start by observing that a modular support $\sigma: R^n \to \N^u$ that takes values in $\{0,1\}^u$ allows us to associate a matroid to a code $C$. More precisely, the minimal ones among the supports of the codewords of $C$ are the circuits of a matroid. This generalizes the well-known fact that one may associate to a linear block-code the matroid represented by its parity-check matrix, whose circuits correspond to the minimal supports of the nonzero codewords of $C$ with respect to the Hamming weight.

\begin{theorem}
Let $R$ be a finite commutative ring and let $\sigma:R^n \to \{0,1\}^u$ be a modular support. Let $0\neq C \subseteq R^n$ be a code. Then the elements of the set
$$\mC:=\{\sigma(v) \st v \in\Min(C)\}$$
are the circuits of a matroid.
\end{theorem}

\begin{proof}
If a support $\sigma$ takes values in $\{0,1\}^u$, then the support of a vector can be naturally identified with a subset of $[u]$. In order to show that $\mC$ is the set of circuits of a matroid, we check the circuit axioms as stated in~\cite[page 9]{oxley}.

Properties (C1) and~(C2) are immediate to verify.
To see that Property (C3) holds, suppose that $\sigma(v),\sigma(w)\in \mC$, that $\sigma(v)\neq\sigma(w)$, and that $(\sigma(v)\wedge\sigma(w))_i \neq 0$.
By repeatedly applying Property~\ref{PP} and up to exchanging the role of $v$ and $w$, one sees that there exists $r \in R$ with $\sigma(v-rw)_i=0$. We claim that $v-rw \neq 0$. Indeed, if $v=rw$ then we would have $\sigma(v) = \sigma(rw) \le \sigma(w)$. Since $\sigma(w)$ is minimal by assumption and $v \neq 0$, it must be that $\sigma(v)=\sigma(w)$, a contradiction.

Since $v-rw \neq 0$, we have $\sigma(v-rw) \neq 0$. Fix $z\in C$ with $\sigma(z)\in\mC$, $\sigma(z) \le \sigma(v-rw)$. We have
\begin{equation}
\sigma(z) \le \sigma(v-rw) \le \sigma(v) \vee \sigma(-rw) \le \sigma(v) \vee \sigma(w).
\end{equation}
Moreover, $0=\sigma(v-rw)_i \ge \sigma(z)_i$. This establishes Property~(C3).
\end{proof}

We start our study of the minimal codewords by showing that the minimal codewords of a code $C$ coincide with those of its socle. We also show that the minimal codewords are determined by their support, up to multiplication by a unit.

\begin{theorem}\label{lemma:minsupp}
Let $C\subseteq R^n$ be a code and assume that $\sigma$ is modular. The following hold.
\begin{enumerate}[label={(\arabic*)}]
\item \label{lemma:minsupp1} The set of minimal codewords of $C$ is
\begin{eqnarray*}
\Min(C) & = & \bigcup_{i=1}^\ell (0\times\ldots\times 0\times \Min(C_i)\times 0\times\ldots\times 0)  \\
 & \subseteq & \bigcup_{i=1}^\ell (0\times\ldots\times 0\times (0:_{C_i}\mathfrak{M}_i)\times 0\times\ldots\times 0)\;\; \subseteq\;\; 0:_{C}J.
\end{eqnarray*}
\item \label{lemma:minsupp2} In particular, $$\Min(C)=\Min(0:_C J).$$
\item \label{lemma:minsupp3} If $v,w\in\Min(C)$ are minimal codewords with $\sigma(v)=\sigma(w)$, then $v=rw$ for some invertible $r\in R$.
\end{enumerate}
\end{theorem}

\begin{proof}
\begin{enumerate}[label={(\arabic*)}]
\item By Theorem~\ref{thm:supp*}, up to a permutation of the coordinates of $\N^u$, $\sigma$ decomposes as a product $\sigma=\sigma_1\times\ldots\times\sigma_\ell$, where
each $\sigma_i$ is a modular support.
Let $i \in [\ell]$ and $v=(v_1,\ldots,v_\ell)\in\Min(C)$ with $v_i\neq 0$.
Then $0\neq \sigma(e_iv)\leq\sigma(v)$, hence $\sigma(e_iv)=\sigma(v)$. In particular, $\sigma_j(v_j)=0$ for all $j\neq i$, hence $v_j=0$ for all $j\neq i$. Therefore, $v=e_iv$ and $v_i\in\Min(C_i)$. This proves the equality in the statement. 

Suppose now that $(R,\mathfrak{M})$ is a finite local ring and $v\in\Min(C)$. If $r\in R$ is such that $rv\neq 0$, then $\sigma(v)=\sigma(rv)$. Since  $\sigma$ is modular, there exists $s\in R$ such that $\sigma(v-srv)<\sigma(v)$, hence $v-srv=0$ by the minimality of $\sigma(v)$. Hence $1-sr\in 0:_{R} v\subseteq\mathfrak{M}$. This shows that $sr\not\in\mathfrak{M}$, hence $r\not\in\mathfrak{M}$. Therefore, $v\in 0:_C\mathfrak{M}$, which proves the second inclusion. 

The inclusion follows from the fact that 
$0:_C J=(0:_{C_1}\mathfrak{M}_1)\times\ldots\times (0:_{C_\ell}\mathfrak{M}_\ell)$.

\item  This follows from part~\ref{lemma:minsupp1}, since $0:_{C}J\subseteq C$ implies $$\Min(0:_{C}J)\supseteq \Min(C)\cap (0:_{C}J)=\Min(C),$$ where the equality follows from part (1). Conversely, if $v\in\Min(0:_{C}J)$, then there exists $w\in\Min(C)$ such that $\sigma(w)\leq\sigma(v)$. Since $w\in 0:_{C}J$ by part (1), then $\sigma(w)=\sigma(v)$ and $v\in\Min(C)$.

\item Since $\sigma$ is modular, there exists $r\in R$ such that $\sigma(v-rw)<\sigma(v)$. By the minimality of $\sigma(v)$, $v-rw=0$, hence $v=rw$. Exchanging the roles of $v$ and $w$ one sees that there exists $s\in R$ such that $w=sv$. Therefore, $(1-rs)v=0$, so $1-rs\in 0:_{R}v$. By part~\ref{lemma:minsupp1}, $v=e_iv$ for some $i\in[\ell]$ and  $0:_{R}v=R_1\times\ldots\times R_{i-1}\times\mathfrak{M}_i\times R_{i+1}\times\ldots\times R_\ell$. Since $1-s_ir_i\in\mathfrak{M}_i$,
then $r_i\not\in\mathfrak{M}_i$, hence $r_i$ is invertible. Let $\overline{r}=(1,\ldots,1,r_i,1,\ldots,1)$. Then $v=\overline{r}w$ and $\overline{r}\in R$ is invertible. \qedhere
\end{enumerate}
\end{proof}

Theorem~\ref{lemma:minsupp} implies that a code  generated by its minimal codewords must be a subcode of $0:_{R^n}J$. In the next theorem we prove that every subcode of $0:_{R^n}J$ is generated by its minimal codewords.

\begin{thm}\label{thm:mingens}
Let $0\neq C\subseteq R^n$ be a code and assume that~$\sigma$ is modular. Then $0:_C J$ has a minimal system of generators consisting of codewords that are minimal in $C$. Moreover, every minimal system of generators of $0:_C J$ consisting of minimal codewords has the same cardinality $M(0:_C J)$.
In particular, $C$ has a minimal system of generators consisting of minimal codewords if and only if $C\subseteq 0:_{R^n}J$. If this is the case, then every such minimal system of generators has cardinality $M(C)$.
\end{thm}

\begin{proof}
By Theorem~\ref{lemma:minsupp}\ref{lemma:minsupp2} we have $\Min(C)=\Min(0:_C J)$. 
Let $D=0:_C J$. Since every system of generators of $D$ contains a minimal one, in order to show that $D$ has a minimal system of generators consisting of minimal codewords, it suffices to show that the elements of~$\Min(D)$ generate $D$.

Let $v\in D$ and suppose by contradiction that $v$ is a codeword of minimal support among those in the set $D\setminus\langle\Min(D)\rangle$. Since $v\not\in\Min(D)$, then there is a $w\in\Min(D)$ such that $\sigma(w)<\sigma(v)$. Let $i\in[u]$ such that $\sigma(w)_i\neq 0$. 
By Theorem~\ref{lemma:minsupp}\ref{lemma:minsupp1}, $w=e_jw$ for some $j\in[\ell]$ and $w_j\mathfrak{M}_j=0$. By Property~\ref{PP} there exists $r\in R$ such that $\sigma(w-rv)_i=0$. Since 
\begin{equation}\label{eqn:rj}
0=\sigma(w-rv)_i=\sigma(e_jw-e_jrv)_i<\sigma(e_jw)_i=\sigma(w)_i,   
\end{equation}
then $r_j\not\in\mathfrak{M}_j$. Indeed, $v\in 0:_C J$ implies that $v_j\in 0:_{C_j}\mathfrak{M}_j$. Hence, if $r_j\in\mathfrak{M}_j$, then $e_jw-e_jrv=e_jw$, contradictiong equation (\ref{eqn:rj}).
Let $s=(1,\ldots,1,r_j,1,\ldots,1)\in R$. Then $s\in R$ is invertible and $$\sigma(w-sv)_i=\sigma(e_jw-e_jsv)_i=\sigma(e_jw-e_jrv)_i=\sigma(w-rv)_i=0.$$
Then $\sigma(w-sv)<\sigma(v)$, hence  by the minimality of $\sigma(v)$ among the supports of elements of $D\setminus\langle\Min(D)\rangle$ we have that $w-sv\in\langle\Min(D)\rangle$. Since $s$ is invertible, this implies that $v\in\langle\Min(D)\rangle$, which contradicts the assumption that $v\in D\setminus\langle\Min(D)\rangle$.

In order to prove that every minimal system of generators of $D$ consisting of minimal codewords has cardinality $M(D)$, write $D=D_1\times\ldots\times D_\ell$. By Theorem~\ref{lemma:minsupp}\ref{lemma:minsupp1}, every minimal codeword $v$ of $D$ satisfies $v=e_iv$ for some $i\in[\ell]$. Therefore, each minimal system of generators of $D$ consisting of minimal codewords is the union for $i\in[\ell]$ of minimal systems of generators of $0\times\ldots\times 0\times D_i\times 0\times\ldots\times 0$. Since $R_i$ is a finite local ring, the cardinality of any minimal system of generators of $0\times\ldots\times 0\times D_i\times 0\times\ldots\times 0$ is $\mu(D_i)$ by Theorem~\ref{matsumura}. Therefore, the cardinality of a minimal system of generators of $D$ consisting of minimal codewords is $\mu(D_1)=\ldots+\mu(D_\ell)=M(D)$.
\end{proof}

We stress that not every minimal system of generators of a code $C\subseteq 0:_{R^n}J$ consists of minimal codewords.

\begin{example}
The element $(2,3)$ is not an element of minimal support in $C=\langle(2,3)\rangle\subseteq\Z_6^2$. However, $(2,0),(0,3)$ are elements of minimal support that generate $C$. Here $C_1=\langle(2,0)\rangle\subseteq\Z_3^2$, $C_2=\langle(0,1)\rangle\subseteq\Z_2^2$, and $M(C)=\mu(C_1)+\mu(C_2)=2$.
\end{example}

The following property of minimal codewords will be needed in the next proposition. 

\begin{lemma}\label{lem:replace} 
Let $\sigma$ be a modular support on $R^n$.
Let $C \subseteq 0:_{R^n}J$ be a code and let $v \in C$ with $\sigma(v)_i \neq 0$. Then there exists $w \in\Min(C)$ with $\sigma(w) \le \sigma(v)$ and $\sigma(w)_i \neq 0$.
\end{lemma}

\begin{proof}
Write $C=C_1\times\ldots\times C_\ell$.
By Theorem \ref{thm:supp*}, we may assume without loss of generality that $(R,\mathfrak{M})$ is a finite local ring. Indeed, if $\sigma=\sigma_1\times\ldots\sigma_{\ell}$ and $\sigma(v)_i=\sigma_k(v_k)_i$ for some $k\in[\ell]$, then $e_kv\in 0\times\ldots\times 0\times R_k\times 0\times\ldots\times 0$ has $\sigma(e_kv)\leq\sigma(v)$ and $\sigma(e_kv)_i\neq 0$. If $w_k\in \Min(C_k)$ and $\sigma_k(w_k)_i\neq 0$, then $\sigma_k(w_k)\leq\sigma_k(v_k)$, therefore $$\sigma(0,\ldots,0,w_k,0,\ldots,0)\leq\sigma(e_kv)\leq\sigma(v)$$
and $\sigma(0,\ldots,0,w_k,0,\ldots,0)_i=\sigma_k(w_k)_i\neq 0$.
Moreover, $(0,\ldots,0,w_k,0,\ldots,0)\in\Min(C)$ as $w_k\in\Min(C_k)$.

Proceed by induction on $|\sigma(v)|$. Let $v'\in\Min(C)$ with $\sigma(v') \le \sigma(v)$. If $\sigma(v')_i \ne 0$ then let $w=v'$, else fix $j \in [u]$ with $\sigma(v')_j \neq 0$. By Corollary \ref{lemadd} there exists $r\in R$ invertible such that $\sigma(v'-rv)_j=0$. Hence $\sigma(v'-rv)<\sigma(v)$ and $\sigma(v'-rv)_i\neq 0$, by Lemma~\ref{prsupp}\ref{prsupp2}. 
So we may apply the induction hypothesis to $v'-rv$ and obtain $w\in\Min(C)$ such that $\sigma(w)\leq\sigma(v'-rv)\leq\sigma(v)$ and $\sigma(w)_i\neq 0$.
\end{proof}

We now prove that modularity allows us to produce minimal system of generators of submodules of $0:_{R^n}J$, whose supports have a shape which is reminiscent of the rows of a matrix in reduced row-echelon form.

\begin{proposition}\label{special_gens} 
Let $j\geq 1$ and let $C\in\mathcal{S}_j(0:_{R^n}J)$.
If $\sigma$ is modular, then $C$ has a minimal system of generators $\{v_1,\ldots,v_j\}$ such that for all $i \in [j]$ there exists $k_i\in [u]$ with $\sigma(v_i)_{k_i} \neq 0$ and $\sigma(v_h)_{k_i}=0$ for all $h \neq i$. 
\end{proposition}

\begin{proof}
Any system of generators with the required property is minimal, since  for all $i\in[j]$ $$\sigma(v_i)\not\leq\bigvee_{h\neq i}\sigma(v_h).$$ We prove that $C$ has such a system of generators by induction on $j$. The statement is trivial if $j=1$. Hence assume $j\ge 2$ and fix a minimal system of generators $\{w_1,\ldots,w_j\}$ of~$C$. Up to permuting the entries in the supports, we may assume without loss of generality that $\sigma(w_1)_1 \neq 0$. By Corollary~\ref{lemadd} there exist $r_2,\ldots,r_j \in R$ with $\sigma(w_i-r_iw_1)_1=0$ for all $i \in \{2,\ldots,j\}$. 
Let $w'_i:=w_i-r_iw_1$ for $i \in \{2,\ldots,j\}$ and observe that $C$ is generated by $\{w_1,w_2',\ldots,w_j'\}$. Moreover, $\sigma(v)_1=0$ for all $v\in C'=\langle w_2',\ldots,w_j'\rangle$ by Lemma~\ref{prsupp}\ref{prsupp3}. We apply the induction hypothesis to the code $C'=\langle w_2',\ldots,w_j'\rangle$, obtaining a system of generators $\{v_2,\ldots,v_j\}$ of $C'$ such that for all $i \in \{2,\ldots,j\}$ there exists $k_i$ with $\sigma(v_i)_{k_i} \neq 0$ and $\sigma(v_h)_{k_i}=0$ for $h \in \{2,\ldots,j\} \setminus \{i\}$. By Corollary~\ref{lemadd} we find $r_2',\ldots,r_j' \in R$ with $\sigma(w_1-r_i'v_i)_{k_i}=0$ for $i \in \{2,\ldots,j\}$. Finally, let $v_1:=w_1-\sum_{i=2}^j r_i'v_i$ and set $k_1=1$. By parts \ref{prsupp2} and \ref{prsupp3} of
Lemma~\ref{prsupp} we have $\sigma(v_1)_{k_1}\neq 0$ and $\sigma(v_1)_{k_i}=0$ for $i \in \{2,\ldots,j\}$.
In addition, $\{v_1,\ldots,v_j\}$ is a system of generators of $C$, since $\{w_1,v_2,\ldots,v_j\}$ is. 
\end{proof}

The proposition also implies that the codomain of a modular support cannot be too small.

\begin{corollary}
If $\sigma: R^n \to \N^u$ is modular, then $u \geq M(0:_{R^n}J)$. In particular, if $R$ is a PIR or $JR^n=0$, then $u\geq \ell n$.
\end{corollary}

\begin{proof}
Proposition \ref{special_gens} for $C=0:_{R^n}J$ and $j=M(0:_{R^n}J)$ implies that $u\geq M(0:_{R^n}J)$. If $JR^n=0$, then $R^n=0:_{R^n}J$, hence $M(0:_{R^n}J)=M(R^n)$. If $R$ is a PIR, then $M(0:_{R^n}J)=M(R^n)$ by Proposition \ref{ann}. In both cases, one has that $M(0:_{R^n}J)=\ell n$.
\end{proof}

Understanding the subcodes of $C$ generated by minimal codewords allows us to prove that the generalized of weights of $C$ are attained by subcodes of $0:_C J$. In particular, $C$ and its socle have the same generalized weights. 

\begin{prop}\label{prop:replacew/ann}
Let $R$ be a PIR. Suppose that $\sigma$ is modular and 
let $C\subseteq R^n$ be a code. Let $r\in [M(C)]$ and $D\in\mathcal{S}_j(C)$, $j\geq r$, be such that $d_r(C)=|\sigma(D)|$. Then $0:_D J\in\mathcal{S}_i(C)$ for some $i\geq r$ and $d_r(C)=|\sigma(0:_D J)|.$ In particular, $$d_r(C)=d_r(0:_C J).$$
\end{prop}

\begin{proof}
By Proposition \ref{ann} and Theorem \ref{thm:mingens}, $0:_D J\subseteq D$ is minimally generated by a set of $M(0:_D J)=M(D)$ codewords. Therefore $0:_D J\in\mathcal{S}_{M(D)}(C)$ and $M(D)\geq r$.
Moreover $$|\sigma(0:_D J)|\leq|\sigma(D)|=d_r(C),$$ hence equality holds.
Since $0:_D J\subseteq 0:_C J$, then $$d_r(0:_C J)\leq |\sigma(0:_D J)|=d_r(C).$$ The reverse inequality follows from the inclusion $0:_C J\subseteq C$.
\end{proof}

In particular, this allows us to determine the last generalized weight of $C$.

\begin{corollary}\label{lastgenwgt}
Let $C \subseteq R^n$ be a code and $\sigma$ be a modular support. Assume that either $C \subseteq 0:_{R^n} J$ or $R$ is a PIR. Then
$$d_{M(C)}(C)=|\sigma(0:_C J)|.$$
\end{corollary}

\begin{proof}
We claim that $M(C)=M(0:_C J)$ and $d_{M(C)}(C)=d_{M(C)}(0:_C J)$.
This is clear if $C \subseteq 0:_{R^n} J$, since $C=0:_C J$. 
If $R$ is a PIR, the claim follows from Proposition~\ref{ann} and Proposition~\ref{prop:replacew/ann}.

By Proposition~\ref{lem:equal}, $M(D)\leq M(C)$ for every $D\subseteq 0:_C J$ and the only subcode $D\subseteq 0:_C J$ with $M(D)=M(C)$ is $D=0:_C J$. 
Therefore \begin{equation*}
d_{M(C)}(C)=|\sigma(0:_C J)|. \qedhere
\end{equation*}
\end{proof}

For a given code, we can produce subcodes that attain its generalized weights and that are minimally generated by a set of minimal codewords, whose supports have the same reduced shape as in Proposition \ref{special_gens}.
This technical result plays a crucial role in the proof of Theorem~\ref{mainth}.

\begin{theorem}\label{coro:special_gens4}
Let $C\subseteq R^n$ be a code and let $\sigma$ be a modular support. Assume that either $C\subseteq 0:_{R^n} J$ or $R$ is a PIR. Then, for all $r\in [M(C)]$, there exists a subcode $D\subseteq C$ such that:
\begin{enumerate}
\item $d_r(C)=|\sigma(D)|$,
\item $D$ has a minimal system of generators
$\{v_1,\ldots,v_r\}$ such that $v_i\in\Min(C)$ for all $i\in[r]$. Moreover
$$\sigma(v_i) \not\le \bigvee_{j \neq i}\sigma(v_j).$$ 
\end{enumerate}
\end{theorem}

\begin{proof}
If $C\subseteq 0:_{R^n} J$, then let $D'\subseteq C$ such that $D'\in\mathcal{S}_j(C)$, $j\geq r$, and $d_r(C)=|\sigma(D')|$.
If $R$ is a PIR, then by Proposition \ref{prop:replacew/ann} there exist $j\geq r$ and $D'\subseteq 0:_C J$ such that $D'\in\mathcal{S}_j(C)$ and $d_r(C)=|\sigma(D')|$. In both cases, by Proposition~\ref{special_gens}, $D'$ has a minimal system of generators $w_1,\ldots,w_j$ with the following property: For all $i \in [j]$ there exists $k_i \in [u]$ with $\sigma(w_i)_{k_i} \neq 0$ and $\sigma(w_h)_{k_i}=0$ for $h \neq i$. By Lemma~\ref{lem:replace}, for all $i \in [j]$ there exists $v_i \in\Min(C)$ with $\sigma(v_i) \le \sigma(w_i)$ and $\sigma(v_i)_{k_i} \neq 0$.
In particular, $\sigma(v_i)\not\leq\vee_{h\neq i}\sigma(v_h)$ for $i\in[j]$. Let $D=\langle v_1,\ldots,v_r\rangle$. Notice that $D\in\mathcal{S}_r(C)$, since $v_h\not\in\langle v_k\st k\in[r], k\neq h\rangle$ for all $h\in[r]$. Moreover, \begin{equation*}
    |\sigma(D)|=\bigvee_{i=1}^r\sigma(v_i)\leq \bigvee_{i=1}^j\sigma(w_i)=|\sigma(D')|.
\end{equation*}
Therefore $|\sigma(D)|=|\sigma(D')|=d_r(C)$.
\end{proof}

\begin{notation}
Let $C\subseteq R^n$ be a code and let $j\in[M(C)]$. We let
$$\mathcal{M}_{j}(C):=\{D\subseteq C \st \mbox{$D$ has a minimal system of generators  of $j$ minimal codewords}\}.$$
\end{notation}

Theorem \ref{coro:special_gens4} shows that for all $r\in [M(C)]$ there exists $D\in\mathcal{M}_r(C)$ such that $d_r(C)=|\sigma(D)|$.
In particular, we have shown the following.

\begin{corollary}\label{min_subcodes}
Let $C \subseteq R^n$ be a code and let $\sigma$ be a modular support. Assume that either $C\subseteq 0:_{R^n} J$ or $R$ is a PIR. The following quantities are equal to the $r$-th generalized weight $d_r(C)$, for any $r\in [M(C)]$:
\begin{enumerate}
\item $\min\{|\sigma(D)|\st D\in\mathcal{S}_j(C) \mbox{ for some } j\geq r\}$,
\item $\min\{|\sigma(D)|\st D\in\mathcal{S}_r(C)\}$,
\item $\min\{|\sigma(D)|\st D\in\mathcal{M}_j(C) \mbox{ for some } j\geq r\}$,
\item $\min\{|\sigma(D)| \st D\in\mathcal{M}_r(C)\}.$
\end{enumerate}
\end{corollary}

\begin{proof}
Equality between $d_r(C)$ and 1. holds by definition.
Equality between $d_r(C)$ and 4. follows directly from Theorem \ref{coro:special_gens4}. Equality between $d_r(C)$ and 3. then follows from the chain of inclusions $$\mathcal{M}_r(C)\subseteq\cup_{j\geq r}\mathcal{M}_j(C)\subseteq\cup_{j\geq r}\mathcal{S}_j(C).$$ Similarly, equality between $d_r(C)$ and 2. follows from the chain of inclusions \begin{equation*}
    \mathcal{M}_r(C)\subseteq\mathcal{S}_r(C)\subseteq\cup_{j\geq r}\mathcal{S}_j(C). \qedhere
\end{equation*}
\end{proof}

\begin{remark}
Theorem \ref{thm:mingens} and Corollary \ref{min_subcodes} are in general false for supports that are not modular. For instance, the support of Example~\ref{expato} violates both results taking $C=\F_2^2$.
\end{remark}

We can now prove that the generalized weights form a strictly increasing sequence. This extends a classical result by Wei \cite{W91}.

\begin{theorem}\label{monoton}
Let $C\subseteq R^n$ be a code and let $\sigma$ be a modular support. Assume that either $C\subseteq 0:_{R^n} J$ or $R$ is a PIR. Then
$$\min\wt(C)=d_1(C)<d_2(C)<\ldots<d_{M(C)}(C)=|\sigma(0:_C J)|.$$
\end{theorem}

\begin{proof}
By Corollary~\ref{min_subcodes} and Theorem \ref{lemma:minsupp}\ref{lemma:minsupp1} we may assume without loss of generality that $C \subseteq 0:_{R^n} J$. Let $r\in[M(C)-1]$ and let $D\subseteq C$ be such that $|\sigma(D)|=d_{r+1}(C)$. We may assume that $D$ has a minimal system of generators $\{v_1,\ldots,v_{j+1}\}$ as in Proposition~\ref{special_gens} with $j\geq r$. Then $D':=\langle v_1,\ldots,v_j\rangle\in\mathcal{S}_j(D)$. We have $\sigma(D')  \le \sigma(D)$ and $\sigma(D')_{k_{j+1}}=0<\sigma(D)_{k_{j+1}}$, hence $|\sigma(D')| < |\sigma(D)|$.
In particular, $$d_r(C) \le |\sigma(D')| < |\sigma(D)|=d_{r+1}(C).$$
The two equalities in the statement follow from Lemma~\ref{lem:genwgt}\ref{lem:genwgt1} and 
Corollary~\ref{lastgenwgt}.
\end{proof}

As an application of Theorem~\ref{monoton}, we extend the generalized Singleton bound~\cite[Corollary~1]{W91} to every code over a PIR and some codes over finite commutative rings.

\begin{corollary}\label{singleton}
Let $C \subseteq R^n$ be a code and let $\sigma$ be a modular support. Assume that either $C\subseteq 0:_{R^n} J$ or $R$ is a PIR. Then $$\min\wt(C)+r-1\leq d_r(C)\leq |\sigma(0:_{R^n}J)|-M(C)+r$$
for all $r\in[M(C)]$. In particular, 
$$\min\wt(C)\leq|\sigma(0:_{R^n} J)|-M(C)+1.$$
\end{corollary}

\begin{proof}
The result follows by combining Theorem~\ref{monoton} with $$d_M(C)=|\sigma(0:_{C} J)|\leq|\sigma(0:_{R^n}J)|,$$ where the equality on the left hand side follows from Corollary~\ref{lastgenwgt}.
\end{proof}

The next corollary proves that, for modular supports, any subcode $D$ of $C$ with $d_r(C)=|\sigma(D)|$ has a minimal system of generators consisting of $r$ elements, and no minimal system of generators of larger cardinality. This result allows us to  restrict to such subcodes when studying the generalized weights of $C$.

\begin{corollary}\label{cor:i=MD}
Let $C \subseteq R^n$ be a code and let $\sigma$ be a modular support. Assume that either $C\subseteq 0:_{R^n} J$ or $R$ is a PIR. Let $r\in [M(C)]$ and $D\in\mathcal{S}_j(C)$, $j\geq r$, be such that $d_r(C)=|\sigma(D)|$. 
Then $r=j=M(D)$. 
\end{corollary}

\begin{proof}
Since $D\in\mathcal{S}_j(C)$, then $r\leq j\leq M(D)$ and $D\in\mathcal{S}_{M(D)}(C)$ by Proposition \ref{lem:minspecial}. Then $$|\sigma(D)|\geq d_{M(D)}(C)\geq d_r(C)=|\sigma(D)|,$$ where the first inequality follows from $D\in\mathcal{S}_{M(D)}(C)$ and the second from Lemma \ref{lem:genwgt}\ref{lem:genwgt2a}. Therefore the inequalities are equalities and $r=j=M(D)$ by Theorem \ref{monoton}. 
\end{proof}

In the next theorem, we establish some additional properties of the subcodes of $C$ that realize the generalized weights of $C$.

\begin{theorem}\label{coro:special_gens}
Let $C \subseteq R^n$ be a code and let $\sigma$ be a modular support. Assume that either $C\subseteq 0:_{R^n} J$ or $R$ is a PIR. Let $r\in [M(C)]$ and $D\in\mathcal{S}_r(C)$ be such that $d_r(C)=|\sigma(D)|$. The following hold.
\begin{enumerate}[label={(\arabic*)}]
\item \label{coro:special_gens1} If $v\in\Min(C)$ satisfies $\sigma(v)\leq\sigma(D)$, then $v\in D$. In particular, $\Min(D)=\Min(C)\cap D$.
\item \label{coro:special_gens2} $0:_D J=\langle v\in C\st v\in\Min(C),\, \sigma(v)\leq\sigma(D)\rangle$. 
\end{enumerate}
In particular, if $D\in\mathcal{M}_r(C)$, then $D=\langle v\in C\st v\in\Min(C),\, \sigma(v)\leq\sigma(D)\rangle$.
\end{theorem}

\begin{proof}
\begin{enumerate}[label={(\arabic*)}]
\item If $C\subseteq 0:_{R^n} J$, then also $D\subseteq 0:_{R^n} J$. If $R$ is a PIR, then $0:_D J\subseteq D$ has $\sigma(0:_D J)=\sigma(D)$ by Proposition \ref{prop:replacew/ann}. In both cases, it suffices to prove the thesis under the assumption that $D\subseteq 0:_{R^n} J$.

If $v\not\in D$, consider $D\subsetneq D'=D+\langle v\rangle\subseteq 0:_{R^n}J$. We have $$r=M(D)<M(D')\leq M(D)+1=r+1,$$ where the equalities follows from Corollary \ref{cor:i=MD} and the first inequality follows from Proposition~\ref{lem:equal}. The second inequality follows from observing that, if $D=D_1\times\ldots\times D_\ell$ and $v=e_jv$, then $$D'=D_1\times\ldots\times D_{j-1}\times (D_j+\langle v_j\rangle)\times D_{j+1}\times\ldots\times D_\ell.$$
Since $v_j\not\in D_j$ and $0:_{R_j}\mathfrak{M}_j\supseteq D_j\cup\{v_j\}$, then $$\mu(D_j+\langle v_j\rangle)=\dim_{R_j/\mathfrak{M}_j}(D_j+\langle v_j\rangle)=\dim_{R_j/\mathfrak{M}_j}(D_j)+1=\mu(D_j)+1.$$ Therefore $M(D')=r+1$.
Since $\sigma(D')=\sigma(v)\vee\sigma(D)=\sigma(D)$, then 
$d_{r+1}(C)\leq |\sigma(D)|=d_{r}(C)$, contradicting Theorem \ref{monoton}. It follows that $v\in D$, as desired.

In order to prove that $\Min(D)=\Min(C)\cap D$, it suffices to prove that $\Min(D)\subseteq\Min(C)$. Let $w\in\Min(D)\subseteq C$, then there is a $v\in\Min(C)$ such that $\sigma(v)\leq\sigma(w)\leq\sigma(D)$, where the second inequality follows from $w\in D$. By the first part of the proof, $v\in D$. Therefore $\sigma(w)=\sigma(v)$ and $w\in\Min(C)$.

\item By Theorem \ref{lemma:minsupp}\ref{lemma:minsupp2} and Theorem~\ref{thm:mingens} and part~\ref{coro:special_gens1},
\begin{equation*}
    0:_D J=\langle v\in \Min(D)\rangle=\langle v\in \Min(C)\cap D\rangle=\langle v\in \Min(C)\st \sigma(v)\leq\sigma(D)\rangle. \qedhere
\end{equation*}
\end{enumerate}
\end{proof}

The next result relates the generalized weights of $C$ with those of its factors.

\begin{corollary}
Let $C=C_1\times\ldots\times C_\ell \subseteq R^n$ be a code and let $\sigma$ be a modular support. Assume that either $C\subseteq 0:_{R^n} J$ or $R$ is a PIR. Then
$$d_r(C)=\min\left\{\sum_{j=1}^\ell d_{r_j}(C_j) \st r_1+\ldots+r_\ell=r, \; r_j\in\{0,\ldots,\mu(C_j)\}\right\}$$ for all $r\in[M(C)]$.
\end{corollary}

\begin{proof}
By Theorem \ref{thm:supp*}, up to a permutation of the coordinates of $\N^u$ we can write $\sigma=\sigma_1\times\ldots\times\sigma_\ell$, where $\sigma_j:R_j^n\to\N^{u_j}$ is a modular support for $j\in[\ell]$ and $u=u_1+\ldots+u_\ell$. Fix $r\in[M(C)]$ and $r_1,\ldots,r_\ell$ such that $r=r_1+\ldots+r_\ell$ and $r_j\in\{0,\ldots,\mu(C_j)\}$. For $j\in[\ell]$, let $D_j\in\mathcal{M}_{r_j}(C_j)$ be such that $|\sigma_j(D_j)|=d_{r_j}(C_j)$. Let $D=D_1\times\ldots\times D_\ell$. Then $D\in\mathcal{M}_r(C)$ has $|\sigma(D)|=d_{r_1}(C_1)+\ldots+d_{r_\ell}(C_\ell)$, proving that
$$d_r(C)\leq\min\left\{\sum_{j=1}^\ell d_{r_j}(C_j) \st  r_1+\ldots+r_\ell=r,  \; r_j\in\{0,\ldots,\mu(C_j)\}\right\}.$$

To prove the reverse inequality, let $D=D_1\times\ldots\times D_\ell\in\mathcal{M}_r(C)$. By Theorem~\ref{lemma:minsupp}\ref{lemma:minsupp1} each of the minimal codewords of $C$, say $v_1,\ldots,v_r$, that minimally generate $D$ belongs to $0\times\ldots\times 0\times D_j\times 0\times\ldots\times 0$ for some $j\in[\ell]$. 
Let $$r_j=|\{v_1,\ldots,v_r\}\cap(0\times\ldots\times 0\times D_j\times 0\times\ldots\times 0)|.$$ Then $r=r_1+\ldots+r_\ell$ and $0\leq r_j\leq\mu(D_j)\leq\mu(C_j)$ for $j\in[\ell]$. Moreover, $$|\sigma(D)|=\sum_{j=1}^\ell|\sigma_j(D_j)|\geq\sum_{j=1}^\ell d_{r_j}(C_j),$$
where the last inequality follows from the fact that $D_j\in\mathcal{M}_{r_j}(C_j)$. By Corollary \ref{cor:i=MD}
\begin{align*}
d_r(C) &= \min\{|\sigma(D)|\st D\in\mathcal{M}_r(C)\} \\ &\geq \min\left\{\sum_{j=1}^\ell d_{r_j}(C_j)\st r=r_1+\ldots+r_\ell,\, r_j\in\{0,\ldots,\mu(C_j)\}\right\}.    \qedhere
\end{align*}
\end{proof}

\section{Codes, Supports, and Monomial Ideals}\label{sec:4}

In this section, we prove that the generalized weights of an $R$-linear code endowed with a modular support are determined by the graded Betti numbers of a monomial ideal associated to the code. We follow the notation of the previous sections.

\begin{notation} In the sequel we work in the multivariate polynomial ring $S=K[x_1,\ldots,x_u]$, where~$K$ is an arbitrary field. 
A \textbf{monomial} of $S$ is a polynomial of the form $x_1^{a_1} \cdots x_u^{a_u}$, where $(a_1,\ldots,a_u) \in \N^u$. In particular, we assume that monomials are monic. A {\bf monomial ideal} is an ideal which has a system of generators consisting of monomials.
\end{notation}

We fix a modular support $\sigma$ on $R^n$. The support $\sigma$ can be used to associate a monomial ideal to a subcode of $R^n$ as follows. 

\begin{defn} \label{def:IC}
For any $v \in R^n\setminus\{0\}$, let 
$$m_v:=x^{\sigma(v)}:=x_1^{\sigma(v)_1}\cdots x_u^{\sigma(v)_u}\in S.$$
For $0\neq C \subseteq R^n$, let 
$$I_C:=\left(m_v\st v\in C\setminus\{0\} \right)\subseteq S.$$
\end{defn}

Notice that not every monomial $m \in I_C$ corresponds to the support of a codeword $v \in C \setminus \{0\}$. 
However, every monomial $m \in I_C$ is of the form $m=m_v \cdot m'$ for some $v \in C \setminus \{0\}$ and some monomial $m' \in S$.

\begin{proposition} \label{prop:minimal}
Let $0\neq C \subseteq R^n$ be a code. Then $$I_C=I_{0:_C J}=\left(m_v \st v\in\Min(C)\right)$$ and $\left\{m_v \st v\in\Min(C)\right)\}$ is a minimal system of generators of $I_C$.
\end{proposition}

The next theorem is the main result of this paper. We prove that the graded Betti numbers of the monomial ideal associated to a code determine its generalized weights.

\begin{theorem} \label{mainth}
Let $\sigma$ be a modular support and let $C\subseteq R^n$ be a code.
Assume that either $C\subseteq 0:_{R^n} J$ or $R$ is a PIR.
Let $I_C\subseteq S$ be the monomial ideal associated to $C$ and let $r\in[M(C)]$. Then $M(C)$ is the projective dimension of $S/I_C$ and $d_r(C)$ is the minimum shift (i.e., the minimum degree of a nonzero element) in the $r$-th free module in a minimal free resolution of $S/I_C$.
In particular, the graded Betti numbers of $S/I_C$ determine $M(C)$ and the generalized weights of $C$.
\end{theorem}

\begin{proof}
Let $I_C=(m_1,\ldots,m_t)$ where $m_1,\ldots,m_t$ are a minimal system of monomial generators of $I_C$. By Theorem~\ref{lemma:minsupp}\ref{lemma:minsupp3} $\langle v\rangle=\langle w\rangle$ if and only if $\sigma(v)=\sigma(w)$ for $v,w\in\Min(C)$. By Theorem~\ref{thm:mingens}, $0:_C J$ has a minimal system of generators consisting of minimal codewords, hence $t\geq M(0:_C J)=M(C)$ by Proposition~\ref{ann}. For each $i\in [t]$ let $v_i\in\Min(C)$ such that $m_{v_i}=m_i$.
For any $A\subseteq[t]$ let $$m_A=\lcm\{m_i\st i\in A\}=x^{\sigma(\langle v_i\st i\in A\rangle)}.$$
A graded free resolution of $S/I_C$ is given by the Taylor complex~\cite[Section~7.1]{MITD}
$$0 \longrightarrow \mathbb{F}_t \stackrel{f_t}{\longrightarrow} \mathbb{F}_{t-1} \stackrel{f_{t-1}}{\longrightarrow}\ldots \stackrel{f_2}{\longrightarrow} \mathbb{F}_1 \stackrel{f_1}{\longrightarrow} S\longrightarrow S/I_C \longrightarrow 0,$$
where $$\mathbb{F}_r=\bigoplus_{A\subseteq [t],\; |A|=r} S(-\deg(m_A))$$ with basis $\{e_A\st A\subseteq[t],\; |A|=r\}$ and 
$$f_r(e_A)=\sum_{k=1}^r (-1)^{k+1}\frac{m_A}{m_{A\setminus\{i_k\}}}e_{A\setminus\{i_k\}}\; \mbox{ for }\; A=\{i_1,\ldots,i_r\}\subseteq [t].$$
The Taylor resolution is in general not minimal: A cancellation occurs between the modules $S(-\deg(m_A))\subseteq\mathbb{F}_r$ and $S(-\deg(m_B))\subseteq\mathbb{F}_{r-1}$ if and only if $B=A\setminus\{k\}$ for some $k\in A$ and $m_A=m_B$. Notice that $m_A=m_B$ if and only if $m_{k}\mid\lcm\{m_i\st i\in B\}$, that is, if and only if $\sigma(v_{k})\leq\vee_{i\in B}\sigma(v_i)$.
 
Let
\begin{equation}\label{mfr}
0 \longrightarrow \mathbb{G}_{p} \stackrel{g_{p}}{\longrightarrow} \mathbb{G}_{p-1} \stackrel{g_{p-1}}{\longrightarrow}\ldots \stackrel{g_2}{\longrightarrow} \mathbb{G}_1 \stackrel{g_1}{\longrightarrow} S \longrightarrow S/I_C \longrightarrow 0
\end{equation}
be a minimal free resolution obtained from the Taylor resolution after making all the possible cancellations. In particular, $p$ is the projective dimension of $S/I_C$. 

For $A\subseteq [t]$ with $|A|=r$, let $$C_A=\langle v_i\st i\in A\rangle\subseteq C.$$ 
Notice that, if $v_k\in\langle v_i\st i\in A\setminus\{k\}\rangle$ for some $k\in A$, then $S(-\deg(m_A))$ cancels with $S(-\deg(m_{A\setminus\{k\}}))$ while passing from the Taylor resolution to resolution~(\ref{mfr}). Therefore, the direct summands $S(-\deg(m_A))$ appearing in~(\ref{mfr}) come from subcodes $C_A\in\mathcal{M}_r(C)$. 
Since $\mathcal{M}_i(C)=0$ for $i>M(C)$, then $p\leq M(C)$.

For $r\in[M(C)]$, let $b_r$ be the smallest shift appearing in the $r$-th module of a minimal graded free resolution of $S/I_C$, i.e. $b_r$ is the smallest degree of a nonzero element of~$\mathbb{G}_r$.
In particular, $b_1$ is the smallest degree of a minimal generator of $I_C$, hence $$b_1=d_1(C).$$

We claim that $b_r=d_r(C)$ for all $r\in[M(C)]$. 
Fix a value of $r$ and choose $A\subseteq[t]$ with $|A|=r$ such that $b_r=\deg(m_A)$. 
Then $\sigma(C_A)=\vee_{i\in A}\sigma(v_i)=b_r$. Since $C_A\in\mathcal{M}_r(C)$, then 
\begin{equation} \label{keyy}
    d_r(C)\le b_r
\end{equation} by Theorem~\ref{min_subcodes}.

To prove the reverse inequality of \eqref{keyy}, we start by observing that 
$$d_r(C)=\min\{|\sigma(C_A)| \ : \ |A|=r,\ C_A\in\mathcal{M}_r(C) \}$$ 
by Theorem~\ref{min_subcodes} and Theorem~\ref{lemma:minsupp}\ref{lemma:minsupp3}. Hence $d_r(C)$ is one of the shifts appearing in the $r$-th free module of the Taylor resolution of $S/I_C$. 
In order to complete the proof, it suffices to prove the following
\begin{claim}\label{claim}
$\mathbb{G}_r$ contains at least one direct summand $S(-d_r(C))$. 
\end{claim}

To prove the claim, suppose that we have made all the possible cancellations until the $(r-1)$-st step of the resolution. Therefore we have a free resolution of the form $$0 \longrightarrow \mathbb{F}_t \stackrel{f_t}{\longrightarrow}\ldots \mathbb{F}_{r+1}
\stackrel{h_{r+1}}{\longrightarrow}
\mathbb{H}_{r} \stackrel{h_{r}}{\longrightarrow} \mathbb{G}_{r-1} \stackrel{g_{r-1}}{\longrightarrow}\ldots
\stackrel{g_2}{\longrightarrow} \mathbb{G}_1 \stackrel{g_1}{\longrightarrow} S \longrightarrow S/I_C \longrightarrow 0.$$
By Theorem~\ref{coro:special_gens4} and Corollary \ref{cor:i=MD}, $\mathbb{H}_{r}$ contains a direct summand $S(-d_r(C))$. Consider now the possible cancellations between $\mathbb{F}_{r+1}$ and $\mathbb{H}_r$. The map \begin{equation}\label{syzr-1}
\mathbb{H}_{r}\stackrel{h_{r}}{\longrightarrow} \mathrm{Syz}_{r-1}(I_C)
\end{equation}
is surjective and corresponds to a choice of generators of
the $(r-1)$-st syzygy module $\mathrm{Syz}_{r-1}(I_C)$ of $I_C$. 
A cancellation between $\mathbb{F}_{r+1}$ and $\mathbb{H}_{r}$ comes from an element in the kernel of~$h_{r}$ which has an invertible entry, hence it corresponds to eliminating a non-minimal generator of $\mathrm{Syz}_{r-1}(I_C)$.
Claim~\ref{claim} amounts to showing that, among all direct summands $S(-d_r(C))$ of~$\mathbb{H}_{r}$,  there is at least one which does not cancel with a direct summand of $\mathbb{F}_{r+1}$. If they all cancel, however, $\mathrm{Syz}_{r-1}(I_C)$ has no elements in degree $d_r(C)$. However this is not possible, since the map in~(\ref{syzr-1}) is surjective and no component of $h_{r}$ is the zero map, hence the image of $h_{r}$ contains a nonzero element of degree $d_r(C)$. In particular, $\mathbb{G}_p=S(-|\sigma(0:_C J)|)^s$ for some $s\geq 1$ and $p=M(C)$.
\end{proof}

Since minimal free resolutions of monomial ideals are easy to compute, Theorem \ref{mainth} gives a way to efficiently compute the generalized weights of a code, if one knows the supports of its minimal codewords. Notice however that the problem of computing the supports of the minimal codewords (or just the minimum distance) of a code is NP-hard, even for the special case of binary linear codes endowed with the Hamming distance, see \cite{vardy97}.

We conclude the section with some examples.

\begin{example}
Let $R=\Z_6$ and consider the support $\omega: \Z_6 \to \N^2$ defined as follows:
$$\omega(0)=(0,0), \quad
\omega(1)=(2,1), \quad
\omega(2) = (0,1), \quad 
\omega(3) = (2,0), \quad
\omega(4)=(0,1), \quad 
\omega(5)=(2,1).$$
It can be checked that $\omega$ is modular.
We extend $\omega$ componentwise to a function $\omega: \Z_6^3 \to \N^6$ and compose it with the $\Z_6$-linear map defined by the invertible matrix
$$A=\begin{pmatrix}
 3 & 4 & 1 \\
 5 & 3 & 3 \\
 2 & 4 & 5 
\end{pmatrix}.
$$
In other words, we define $\sigma: \Z_6^3 \to \N^6$ by $\sigma(v) =(\omega((Av)_1),\omega((Av)_2),\omega((Av)_3))$ for all $v \in \Z_6^3$. Then $\sigma$ is modular by Proposition~\ref{newPP}\ref{newPP7}. Let $C=\langle (3,1,2),(2,4,3)\rangle\subseteq \Z_6^3$. Then $C$ has the following six minimal codewords and supports:
 \begin{eqnarray*}
     (3,3,3)  & & (0,0,2,0,2,0), \\
     (0,4,2)   & & (0,0,0,0,0,1), \\
     (2,0,4)    & & (0,1,0,1,0,0), \\
     (3,3,0)    & & (2,0,0,0,0,0), \\
     (0,2,4)   & &  (0,0,0,0,0,1), \\
     (4,0,2)  & & (0,1,0,1,0,0).
\end{eqnarray*}
In particular, the ideal associated to $C$ is $$I_C=(x_1^2,x_2x_4,x_3^2x_5^2,x_6).$$
Notice that the generators of $I_C$ form a regular sequence, therefore a minimal free resolution of $S/I_C$ is given by the Koszul complex. In other words, there is no cancellation in the Taylor complex
$$0\longrightarrow S(-9)\longrightarrow \begin{array}{c} S(-8)\\ \oplus \\ S(-7)^2 \\ \oplus \\ S(-5) \end{array}\longrightarrow \begin{array}{c} S(-6)^2\\ \oplus \\ S(-5)^2 \\ \oplus \\ S(-4)\\ \oplus \\ S(-3)^2 \end{array}
\longrightarrow \begin{array}{c} S(-4)\\ \oplus \\ S(-2)^2 \\ \oplus \\ S(-1) \end{array} \longrightarrow S\longrightarrow S/I_C\longrightarrow 0.$$
Therefore one has $M(C)=4$, $d_1(C)=1$, $d_2(C)=3$, $d_3(C)=5$, and $d_4(C)=9$. Looking at the maps in the minimal free resolutions, one sees that
\begin{align*}
 d_1(C)&=|\sigma(\langle(0,2,4)\rangle)|=|(0,0,0,0,0,1)|=1, \\ d_2(C)&=|\sigma(\langle(0,2,4),(3,3,0)\rangle)|=|\sigma(\langle(0,2,4),(2,0,4)\rangle)|=|(2,0,0,0,0,1)|=|(0,1,0,1,0,1)|=3,\\
 d_3(C)&=|\sigma(\langle(0,2,4),(3,3,0),(2,0,4)\rangle)|=|(2,1,0,1,0,1)|=5, \\
 d_4(C)&=|\sigma(\langle(0,2,4),(3,3,0),(2,0,4),(3,3,3)\rangle)|=|(2,1,2,1,2,1)|=9.
\end{align*}
\end{example}

In the next example, we show that all the cancellations  discussed in the proof of Theorem~\ref{mainth} can actually occur. 

\begin{example}
Let $\sigma$ be the Hamming weight and let $C\subseteq\FF_2^4$ be the even weight code, i.e., the code consisting of all vectors of even weight. The minimal codewords of $C$ are all vectors of weight $2$ of $\FF_2^4$,
therefore the associated ideal $I_C$ is the ideal generated by all squarefree monomials of degree $2$ in $4$ variables. A minimal free resolution of $S/I_C$ is well-known and has the form
$$0\longrightarrow S(-4)^2\longrightarrow S(-3)^8\longrightarrow S(-2)^6\longrightarrow S\longrightarrow S/I_C\longrightarrow 0.$$
The Taylor complex associated to $S/I_C$ is as follows
{\footnotesize $$0\longrightarrow S(-4)\longrightarrow S(-4)^6\longrightarrow S(-4)^{15}\longrightarrow \begin{array}{c} S(-4)^{16} \\ \oplus \\ S(-3)^4\end{array}\longrightarrow \begin{array}{c} S(-4)^3 \\ \oplus \\ S(-3)^{12}\end{array}\longrightarrow S(-2)^6\longrightarrow S\longrightarrow S/I_C\longrightarrow 0.$$}
We use the notation of the proof of Theorem~\ref{mainth} when referring to the free modules in the resolutions above. 
By comparing the two free resolutions, one sees that the modules~$\FF_4$, $\FF_5$ and $\FF_6$ in the Taylor complex completely cancel. Moreover, the free summand~$S(-3)^4$ of~$\FF_3$, corresponding to the minimal shift in its homological degree, cancels, as well as the direct summand $S(-4)^{14}$.  Finally, the direct summands $S(-3)^4$ and $S(-4)^3$ of $\FF_2$ cancel. 
\end{example}

It is easy to find examples of functions that are not supports according to Definition~\ref{defsupp}, and for which the result of Theorem~\ref{mainth} does not hold.

\begin{example}
Let $R=\mathbb{Z}_4$. Let $\wL$ denote the Lee weight and let $\sL$ be its coordinatewise extension to $\mathbb{Z}_4^3$. In Example~\ref{rem:lee} we observed that $\sL$ is not a support according to Definition~\ref{defsupp}. 
Consider the code $C=\langle (1,1,0), (3,2,1)\rangle\subseteq \mathbb{Z}_4^3$. Then $$\Min(C)=\{(1,1,0),(3,3,0),(0,1,3),(0,3,1),(1,0,1),(3,0,3)\}$$
and $I_C=(xy,xz,yz)\subseteq S=K[x,y,z]$. The graded Betti numbers of a minimal free resolution of $S/I_C$ are 
$$0\longrightarrow S(-3)^2\longrightarrow S(-2)^3\longrightarrow S \longrightarrow S/I_C\longrightarrow 0.$$ In particular, $b_1=2$ and $b_2=3$. 
One can check that $d_1(C)=4$ and $d_2(C)=6$.
\end{example}

Notice that Theorem~\ref{mainth} does not hold in general for supports that are not modular, even for linear codes over fields.

\begin{example}
Let $C=\mathbb{F}_2^2$ with the support $\sigma$ defined in Example~\ref{expato}. The associated monomial ideal is $I_{\mathbb{F}_2^2}=(y)\subseteq K[x,y]$, whose minimal free resolution is $$0\longrightarrow S(-1)\longrightarrow S\longrightarrow S/(y)\longrightarrow 0.$$ We have $\mu(\mathbb{F}_2^2)=2$, while the projective dimension of $S/(y)$ is $1$. In particular, the second free module in a minimal free resolution of $S/(y)$ is $0$ and it does not determine $d_2(\mathbb{F}_2^2)=|\sigma(\mathbb{F}_2^2)|=2$.
\end{example}

\section{The Hamming Support}\label{sec:5}

In this section we let $R=\F_q$ and make some observations on the Hamming support $\sH:\F_q^n \to \{0,1\}^n$, see Example~\ref{exasupp}.
For $v \in \F_q^n$, we identify $\sH(v)$ with a subset of $[n]$. 

Several authors study the matroid $M_C$ associated with the parity-check matrix of a code $C \subseteq \F_q^n$, see e.g. \cite{barg,J2013}. The code $C$, the matroid $M_C$, and the corresponding Stanley-Reisner ideal~$I_C$ relate as follows:
\begin{enumerate}[label={(\arabic*)}]
\item Let $A\subseteq [n]$. Then $A\in M_C$ if and only if there is no $v \in C \setminus \{0\}$ with $\sH(v) \subseteq A$.
\item The squarefree monomials in $I_C$ are exactly those of the form $x^A$, where $A\notin M_C$.
\item The circuits of $M_C$ are the supports of the minimal codewords of $C$.
\item The minimal monomial generators of $I_C$  are in one-to-one correspondence with the circuits of the matroid $M_C$.
\end{enumerate}
Therefore,
\begin{equation} \label{caratt1}
M_C=2^{[n]} \setminus \{A\subseteq [n] \st A\supseteq \sigma(v) \mbox{ for some } v \in C \setminus \{0\} \}=2^{[n]} \setminus \{A\subseteq [n] \st x^A\in I_C \}.
\end{equation}

The resolutions of the ideals associated to various classes of codes (most notably, MDS codes, one/two weight codes, Reed-Muller codes) have been studied in
\cite{ghorpade2020pure,ghorpade2021purity}.

The main result of \cite{J2013} is Theorem~2, which shows that the generalized Hamming weights of a code $C \subseteq \F_q^n$ are determined by the graded Betti numbers of $I_C$. 
Our Theorem \ref{mainth} generalizes \cite[Theorem~2]{J2013} with a different and stand-alone proof strategy that does not rely on any matroid theory results. 

A known fact about codes $C \subseteq \F_q^n$ endowed with the Hamming metric is that they are generated by their minimal codewords~\cite[Lemma 2.1.4]{ashikhmin1998minimal}.  This is a special case of Theorem~\ref{thm:mingens} in this paper.

\begin{corollary} \label{coro:miniH}
Every code $0\neq C \subseteq \F_q^n$ is generated by its minimal codewords.
\end{corollary}

It is natural to ask whether a code $C \subseteq \F_q^n$ is also generated by its codewords of maximal support. The answer to this question is negative in general, as the following simple example shows.

\begin{example}
Let $C=\FF_2^2$. The only codeword of maximal support is $(1,1)$ and $\langle(1,1)\rangle\subsetneq C$. 
\end{example}

Although codes $C \subseteq \F_q^n$ are in general not generated by their maximal codewords, for $q$ sufficiently large most codes are. 

\begin{proposition} \label{prop:estimate}
Let $1 \le k \le n$. The proportion of codes in 
$\F_q^n$, within the $k$-dimensional ones, generated by their maximal codewords is at least
$$\frac{(q-1)^n \, (q^k-1) - (q^n-1)q^{k-1}}{(q^n-1)(q^k-q^{k-1}-1)},$$
and the above fraction approaches 1 as $q \to +\infty$.
\end{proposition}

\begin{proof}
We obtain the result follows from lower bounding the number of $k$-dimensional codes 
with $|C \cap A_n| > q^{k-1}$, where $A_n$ is the set of vectors in~$\F_q^n$ of weight $n$. To obtain such a 
lower bound, denote by $\mF$ the set of $k$-dimensional codes in~$\F_q^n$ and consider the sum
$$S:= \sum_{C \in \mF} |C \cap A_n|.$$
We have
$$|\mF|=\qbin{n}{k}{q}.$$
We write $\mF$ as a disjoint union $\mF=\mF' \cup \mF''$, where $\mF'$ is the set of codes $C \in \mF$ 
with $|C \cap A_n| > q^{k-1}$ and $\mF''$ the set of codes  $C \in \mF$ 
with $|C \cap A_n| \le q^{k-1}$. We then have
\begin{eqnarray}
S &=& \sum_{C \in \mF'} |C \cap A_n| + \sum_{C \in \mF''} |C \cap A_n| \nonumber \\ 
&\le& |\mF'| \, (q^k-1) + (|\mF|- |\mF'|) \, q^{k-1} \nonumber  \\
&=& |\mF'| \, (q^k-q^{k-1}-1)+ |\mF| \, q^{k-1}. \label{eq:S1}
\end{eqnarray}
On the other hand, we can rewrite $S$ as
\begin{equation} \label{eq:S2}
S= \sum_{v \in A_n} |\{C \in \mF \st C \ni v\}| = |A_n| \, \qbin{n-1}{k-1}{q}.
\end{equation}
Combining \eqref{eq:S1} with \eqref{eq:S2}, and using the definition of 
$q$-binomial coefficient, we obtain
$$|\mF'|/|\mF| \ge |A_n| \, \frac{q^k-1}{(q^n-1)(q^k-q^{k-1}-1)} - \frac{q^{k-1}}{q^k-q^{k-1}-1},$$
which is the desired estimate, since $|A_n| = (q-1)^n$.
\end{proof}

\bigskip

\bigskip

\bibliographystyle{amsplain}
\bibliography{idealsbib}

\end{document}